\documentclass[prx,10pt,twocolumn,amsmath,amssymb,nofootinbib,superscriptaddress]{revtex4-2}
\usepackage[british]{babel}

\usepackage[dvipsnames]{xcolor}
\definecolor{myblue}{RGB}{34,31,150}

\definecolor{mygreen}{RGB}{34,150,31}

\usepackage[T1]{fontenc}
\usepackage{graphicx}
\usepackage{bm}
\usepackage{verbatim}
\usepackage{times}
\usepackage{url}
\usepackage{multirow}
\usepackage{mathrsfs}
\usepackage{physics}
\usepackage{mathtools}
\usepackage{bbm}
\usepackage{amsthm}

\usepackage[breaklinks=true,colorlinks=true,linkcolor=myblue,urlcolor=myblue,citecolor=myblue]{hyperref}
\usepackage[capitalise]{cleveref}
\usepackage{orcidlink}

\newtheorem{theorem}{Theorem}

\begin{document}

\title{Closed-form Bayesian quantum estimation of Gaussian states}

\author{Edward Gandar\,\orcidlink{0000-0002-4328-1979}}
\email{Contact author: gandaredward@gmail.com}
\affiliation{Department of Physics and Astronomy, University of Exeter, Stocker Road, Exeter EX4 4QL, United Kingdom}
\affiliation{Secci\'{o}n de F\'{i}sica, Facultad de Ciencias, Universidad de La Laguna, 38203 La Laguna, Spain}
\author{Jes\'{u}s Rubio\,\orcidlink{0000-0002-8193-8273}}
\affiliation{Department of Physics and Astronomy, University of Exeter, Stocker Road, Exeter EX4 4QL, United Kingdom}

\begin{abstract}
Bayesian quantum estimation provides a robust formalism for quantum technologies, particularly in scenarios with limited data and minimal prior information.
Yet, its application to continuous-variable systems has remained limited and largely numerical due to the difficulty of the underlying parameter integrals.
Here, we introduce a variational framework that restricts the optimisation over all measurements and estimators to a finite-dimensional subspace, reducing the problem to a linear system of equations with closed-form solutions. 
These solutions have a geometric interpretation as orthogonal projections of the global optimum onto the chosen subspace.
For Gaussian states, a natural subspace is spanned by operators polynomial in the canonical quadratures, whose moments are computable from the first and second moments of the state.
We further derive a necessary and sufficient condition for global optimality, and show that post-processing with the posterior mean, rather than using the estimator obtained from the subspace, improves performance towards the global optimum.
Through single-shot examples, we show that the framework yields experimentally feasible strategies based on Gaussian operations and quadrature measurements that are either optimal or near-optimal.

\end{abstract}

\maketitle

% Body
\section{Introduction}

Quantum estimation theory (QET) underpins quantum metrology~\cite{demkowicz2015quantum, demkowicz2020multiparameter, albarelli2020perspective, barbieri2022optical}, where resources such as squeezing and entanglement enable precision beyond classical limits~\cite{giovanetti2006quantum, giovannetti2011advances}.
Traditionally, QET has been formulated within a frequentist framework~\cite{helstrom1969quantum,holevo1982testing,braunstein1994statistical}, where performance is characterised via the Fisher information and the Cram\'{e}r--Rao bound.
Such approaches often become effective only once the parameter has been sufficiently localised, either through prior information or preliminary measurements, thereby restricting their applicability in genuinely global or finite-sample regimes~\cite{tsang2012zivzakai,rubio2019quantum,rubio2020bayesian,valeri2020experimental,meyer2023quantum,gebhart2024fundamental,fradil2026a,overton2026adaptive}.

The Bayesian formulation provides an alternative in which prior information is incorporated from the outset and performance is defined globally through a loss function averaged over both the prior and measurement outcomes \cite{jaynes2003probability}.
Early works by Personick showed that optimal measurements can, for the mean square loss (MSL), be constructively derived from an operator equation of Lyapunov type~\cite{personick1970efficient,personick1971application}, and more recent developments have exploited this perspective in finite-sample \cite{rubio2019quantum} or single-shot \cite{mashide2002optimal,bernad2018optimal, zhou2023bayesian,zhou2024bayesian} settings, multiparameter schemes \cite{rubio2020bayesian,demkowicz2020multiparameter,sidhu2020geometric, zhang2024bayesian,albarelli2025measurement}, and symmetry-informed metrologies~\cite{rubio2022quantum,rubio2024first,boeyens2025role}.
Despite this progress, Bayesian QET can be more demanding computationally than its local counterpart, in part because optimisation depends on integrated, prior-weighted state operators rather than on local derivatives.

This difficulty is further exacerbated in continuous-variable (CV) systems, which are central to quantum technologies~\cite{ferraro2005gaussian, weedbrook2012gaussian, adesso2014continuous, serafini2017quantum, miao2019quantum}, due to their infinite-dimensional Hilbert spaces.
In this context, Gaussian states are widely used due to their simplicity and experimental accessibility, being fully characterised by first and second moments and generated by quadratic Hamiltonians and linear operations.
They underpin a range of quantum platforms, including implementations in grav\-i\-ta\-tion\-al-wave detectors~\cite{demkowicz2013fundamental,miao2019quantum,jia2024squeezing}.
But while Gaussian metrology has been extensively developed within local QET~\cite{berry2000optimal,monras2006optimal,monras2007optimal,genoni2011optical,monras2013phase,lang2013optimal,gagatsos2016gaussian,branford2019quantum,oh2019optimal,vsafranek2019estimation,bressanini2024multi,chang2025multiparameter,chattopadhyay2025generic,chattopadhyay2025non}, techniques for its development within the Bayesian framework remain fragmented and comparatively unexplored.

Existing Bayesian approaches to Gaussian metrology can be broadly grouped into (semi-)analytical solutions for specific systems \cite{wang2007quantum, olivares2009bayesian, gammelmark2012bayesian, kiilerich2016bayesian, morelli2021bayesian, zhou2023bayesian, zhou2024bayesian}---which often do not readily extend to more general families of states or measurements---or purely numerical approaches~\cite{rubio2019quantum, nichols2019designing, kundu2021machine, fallani2022learning, ravell2025knowledge,jager2026provable}.
This reflects the aforementioned computational bottleneck.
There is therefore a need for a formulation of Bayesian optimisation in CV metrology that enables the problem to be treated in a systematic and, where possible, analytical manner.
    
In this work, we introduce a variational framework for Bayesian estimation of CV systems.
By restricting the operator that encodes the measurement-estimator pair to a finite-dimensional subspace of Hermitian operators, the problem reduces to a linear system of equations whose solution fully determines the optimal strategy within the subspace.
For Gaussian states and subspaces spanned by polynomials in the quadratures, the coefficients of this linear system are determined by the first and second moments of the state.
We derive closed-form expressions for the associated minimum MSL, establish a geometric interpretation as an orthogonal projection of the global optimum, and give a phase-space criterion characterising when this projection is exact, so that the constrained strategy is globally optimal.
We further show that replacing the constrained estimator with the posterior mean (PM) for the same measurement can only improve, or leave unchanged, the MSL, bringing the performance closer to the global optimum.

This paper is organised as follows.
\cref{sec:preliminaries} reviews CV systems, Gaussian states, and Bayesian QET.
\cref{sec:constrained-optimisation} develops the constrained optimisation framework, giving the linear system and its closed-form solution, the measurements and estimators it admits, and its performance relative to the global optimum.
\cref{sec:examples} applies the framework to displacement and squeezing estimation, recovering known optimal strategies and obtaining near-optimal, experimentally feasible solutions in more general settings.
Finally, \cref{sec:discussion} discusses limitations and outlook.

\vfill

\section{Preliminaries}
\label{sec:preliminaries}

We begin by revisiting key aspects of continuous variable systems, Gaussian states, and Bayesian QET, before presenting our main results in \cref{sec:constrained-optimisation}.
Readers familiar with these topics may proceed directly to that section.

\subsection{Continuous variable systems and Gaussian states}

The paradigmatic example of a CV system is a collection of $N$ bosonic modes with quadrature operators \cite{barnett2002methods,schleich2011quantum,scully1997optics,walls2008quantum}
\begin{equation}
    \hat{\mathbf R} = (\hat q_1,\hat p_1,\ldots,\hat q_N,\hat p_N)^\top,
    \quad
    [\hat R_i,\hat R_j] = i \Omega_{ij},
\end{equation}
where
\begin{equation}
    \hat q_k = \frac{\hat a_k+\hat a_k^\dagger}{\sqrt{2}}, 
    \quad 
    \hat p_k = \frac{\hat a_k-\hat a_k^\dagger}{\sqrt{2}i},
    \quad
    \Omega = \bigoplus_{k=1}^N 
    \begin{pmatrix}
        0 & 1\\
        -1 & 0
    \end{pmatrix},
\end{equation}
with \(\Omega\) the symplectic form \cite{weedbrook2012gaussian,ferraro2005gaussian}.

Any density operator $\rho$\footnote{
In the following, we reserve hats for operators directly associated with canonical coordinates.
} acting on this Hilbert space $\mathcal{H}=L^2(\mathbbm{R}^N)$ admits an equivalent representation as a quasi-probability distribution on the $2N$-dimensional phase space \cite{adesso2014continuous,serafini2017quantum,folland2016harmonic}.
A convenient representation is the characteristic function
\begin{equation}
    \label{eq:characteristic_function}
    \chi_\rho(\vb*{\xi}) = \mathrm{Tr}[\rho D(\vb*{\xi})],
\end{equation}
with Weyl displacement operator\footnote{
For a single mode, the displacement operator is usually written in terms of $\hat{a}$ and $\hat{a}^\dagger$ as $\hat{D}(\alpha)=\mathrm{exp}(\alpha \hat{a}^\dagger - \alpha^* \hat{a})$, with $\alpha \in \mathbbm{C}$. 
It is straightforward to verify that this is equivalent to \cref{eq:Weyl_operator} with $\alpha=\frac{1}{\sqrt{2}}(\xi_q + i \xi_p)$ and $\vb*{\xi} \in \mathbbm{R}^{2}$ \cite{serafini2017quantum}.
}
\begin{equation}
    \label{eq:Weyl_operator}
    D(\vb*{\xi})=\mathrm{exp}(i \hat{\vb{R}}^\top \Omega \vb*{\xi})
\end{equation}
and phase-space coordinates $\vb*{\xi} \in \mathbbm{R}^{2N}$. 
The completeness property \cite{serafini2017quantum,folland2016harmonic} of $D(\vb*{\xi})$ implies a bijective mapping $\rho \mapsto \chi_\rho (\vb*{\xi})$, i.e., every density operator corresponds to a unique characteristic function and vice versa.
The Wigner function of $\rho$ is then related to $\chi_\rho(\vb*{\xi})$ via a symplectic Fourier transform \cite{adesso2014continuous,weedbrook2012gaussian}
\begin{equation}
    \label{eq:Wigner_func_characteristic_func_FT}
    W_\rho(\vb{r})=\int \frac{\dd{\vb*{\xi}}}{\left(2\pi\right)^{2N}} e^{-i \vb{r}^\top \Omega \vb*{\xi}} \chi_\rho(\vb*{\xi}),
\end{equation}
with dual phase-space coordinates $\vb{r} \in \mathbbm{R}^{2N}$. 
More generally, $W_X(\vb{r})$ denotes the Weyl symbol of any operator $X$ \cite{folland2016harmonic,de2011symplectic}; for $X=\rho$, this coincides with the Wigner function.

A state $\rho$ is defined to be Gaussian when the Wigner function is a Gaussian function in $\vb{r}$, i.e.,
\begin{equation}
    \label{eq:Gaussian_Wigner_func}
    W_\rho(\vb{r})=\frac{1}{(2\pi)^N \sqrt{\det V}} \exp \left(-\frac{1}{2}(\vb{r}-\overline{\vb{r}})^{T} V^{-1}(\vb{r}-\overline{\vb{r}})\right),
\end{equation}
or equivalently the characteristic function takes the form
\begin{equation}
    \label{eq:characteristic_func_general_Gaussian}
    \chi_\rho(\vb*{\xi})=\exp[ -\frac{1}{2} \vb*{\xi}^\top (\Omega V \Omega^\top) \vb*{\xi} -i(\Omega \overline{\vb{r}})^\top \vb*{\xi}].
\end{equation}
One then sees that a Gaussian state is fully parametrised by its first moments $\overline{\vb{r}}=\mathrm{Tr}\, (\rho\hat{\vb{R}}) \coloneqq \langle{\hat{R}}\rangle_{\rho}$ and covariance matrix $V$ \cite{adesso2014continuous,serafini2017quantum,weedbrook2012gaussian}, with components 
\begin{equation}
    \label{eq:covariance_matrix}
    V_{ij}=\frac{1}{2}\langle{\hat{R}_i \hat{R}_j +\hat{R}_j \hat{R}_i}\rangle_{\rho} -\langle{\hat{R}_i}\rangle_{\rho} \langle{\hat{R}_j}\rangle_{\rho}.
\end{equation}
In moving to Bayesian QET in the next section, we consider a family of states $\rho(\theta)$ encoding an unknown parameter $\theta$.
For Gaussian states, this means $\theta$ is encoded in the mean vector and/or covariance matrix of the Wigner function, $\overline{\vb{r}}(\theta),\, V(\theta)$.

\subsection{Bayesian QET}\label{sec:bayesian-qet}

In Bayesian QET, any \emph{a priori} information about $\theta$ is encoded in a prior distribution $p(\theta)$.
Upon performing a measurement with outcome $x$, $p(\theta)$ is updated to the posterior distribution $p(\theta|x) \propto p(\theta) p(x|\theta)$ via Bayes' rule, using the likelihood $p(x|\theta)$ of obtaining such a result \cite{kay1993fundamentals, jaynes2003probability}.

While the posterior encodes all the information about $\theta$, a point estimator $\tilde\theta(x)$ for the unknown parameter provides a useful summary.
Constructing estimators requires a loss function $\ell(\tilde\theta,\theta)$ assigning a cost to deviations of the estimate from each possible value for the unknown parameter.
For instance, squared and absolute losses yield the posterior mean (PM) and median estimators, respectively \cite{boeyens2021uninformed}.
Then, given a state $\rho(\theta)$, a prior $p(\theta)$, and a positive operator-valued measure (POVM) $M(x)$ satisfying $M(x)\succeq 0$ and $\int \dd{x} M(x) = \mathbbm{1}$, the MSL over all possible $x$ and $\theta$ is 
\begin{equation}
    \label{eq:average_information_error}
    \mathcal{L}=  \int \dd{\theta} \dd{x}  p(\theta) p(x|\theta) \ell[\tilde{\theta}(x),\theta],
\end{equation}
where the likelihood $p(x|\theta)=\mathrm{Tr}[M(x)\rho(\theta)]$ is calculated using the Born rule. 

A broad and unifying class of loss functions is the \emph{location-isomorphic} family \cite{rubio2024first}
\begin{equation}
    \label{eq:quadratic_loss_f}
    \ell(\tilde\theta, \theta)=[f(\tilde\theta)-f(\theta)]^2, 
 \end{equation} 
where $f\hspace{-0.2em}: \theta \mapsto f(\theta)$ maps $\theta$ to a location parameter, i.e., associated with invariance under translations $f(\theta')=f(\theta) + \gamma$, $\gamma \in \mathbbm{R}$ . 
For example, $f(\theta)=\theta$ and $f(\theta)=\log \theta$ are suitable for problems with invariance under translations and rescaling, respectively~\cite{jaynes2003probability, rubio2022quantum}.
Alternatively, loss functions can be derived from information geometry, based on the statistical manifold induced by either \(\rho(\theta)\) or \(p(x|\theta)\) \cite{amari2000methods,braunstein1994statistical,jorgensen2022bayesian, boeyens2025role}.

For \cref{eq:quadratic_loss_f}, the MSL in \cref{eq:average_information_error} can be written as
\begin{equation}
    \label{eq:average_info_error_general_POVM}
    \mathcal{L} = \lambda + \mathrm{Tr}(\rho_0 \mathcal{M}_{2}) - 2 \mathrm{Tr}(\bar{\rho} \mathcal{M}_{1}),
\end{equation}
with \(\lambda \coloneq \int \dd{\theta}  p(\theta) f(\theta)^2\), \(k\)-th operator moment
\begin{equation}
    \label{eq:Mk_operator_moments}
    \mathcal{M}_{k} = \int \dd{x} M(x) f[\tilde{\theta}(x)]^k,
\end{equation}
and zeroth and first state moments
\begin{subequations}
    \label{eq:prior_averaged_state_moments}
\begin{align}
    \rho_0 &:= \int \dd \theta \, p(\theta) \rho(\theta), \\
    \bar{\rho} &:= \int \dd \theta \, p(\theta) \rho(\theta) f(\theta).
\end{align}
\end{subequations}
Note that $\rho_0$ is always a valid density matrix,\footnote{If, for each $\theta$, \(\rho(\theta) \succeq 0\), \(\mathrm{Tr}[\rho(\theta) ] = 1\), $p(\theta) \geq 0$ and $\int \dd{\theta} p(\theta) =1$, then, for any $\ket{\psi}$, we have $\mel{\psi}{\rho_0}{\psi}=\int \dd{\theta}p(\theta) \mel{\psi}{\rho(\theta)}{\psi} \geq 0$ and $\mathrm{Tr}(\rho_0) = \int \dd{\theta} p(\theta) \mathrm{Tr}[\rho(\theta)]=1$.
Therefore, $\rho_0$ is always a valid density matrix.} since it is a convex combination of other valid density matrices $\rho(\theta)$, but is not guaranteed to be Gaussian.\footnote{$\rho_0$ is Gaussian if the Weyl symbol $W_{\rho_0}(\vb{r})=\int \dd{\theta} p(\theta) W_{\rho(\theta)}(\vb{r})$ is a Gaussian function in $\vb{r}$. While this is not satisfied in general, this is true if the prior is Gaussian, the mean is linear $\bar{\vb*{r}}(\theta)=\bar{\vb*{r}}_0 + \theta \vb{d}$, and the covariance matrix $V(\theta)=V_0$ is constant (as we will see in \cref{sec:1d_displacement}).} 

Jensen's operator inequality $\mathcal{M}_{2}- \mathcal{M}_{1}^2 \succeq 0$ leads to the lower bound 
\begin{equation}
    \label{eq:average_info_error_projective_POVM}
    \mathcal{L} \geq  \lambda + \mathrm{Tr}(\rho_0 \mathcal{M}^2_{1}) - 2 \mathrm{Tr}(\bar{\rho} \mathcal{M}_{1})
    \eqqcolon \mathcal{L}(\mathcal{M}_{1}),
\end{equation}
which is saturated with a projection-valued measure (PVM) \(M(x) = \mathcal{P}(x)\), where $\mathcal{P}(x)\mathcal{P}(x') \to  \delta(x-x')\mathcal{P}(x')$, and the operator $\mathcal{S}$ defined in \cref{eq:lyapunov} below provides the optimal $\mathcal{M}_1$ achieving this bound~\cite{macieszczak2014bayesian,rubio2022quantum}.

Since \cref{eq:average_info_error_projective_POVM} depends only on the single Hermitian operator $\mathcal{M}_{1}$, the minimum MSL over all POVMs and estimators, for a fixed prior $p(\theta)$, state $\rho(\theta)$ and loss function, can be straightforwardly calculated to be 
\begin{equation}
    \label{eq:MSL_min_global}
    \mathcal{L}(\mathcal{S}) \coloneq \lambda - \mathrm{Tr}(\rho_0 \mathcal{S}^2),
\end{equation}
using variational calculus \cite{personick1971application, rubio2024first}. 
As seen, the minimum MSL is the difference between $\lambda$, determined by the prior and loss function, and the information pseudo-gain $\mathrm{Tr}(\rho_0 \mathcal{S}^2)$ \cite{albarelli2025measurement}, which depends, in addition, on the state.
Here, $\mathcal{S}$ solves the Lyapunov equation
\begin{equation}
    \label{eq:lyapunov}
    \mathcal{S}\rho_0 + \rho_0 \mathcal{S} =2 \bar{\rho},
\end{equation}
and is known as the \emph{symmetric posterior mean (SPM) operator}, in analogy with the symmetric logarithmic derivative (SLD) in local QET \cite{albarelli2025measurement, demkowicz2020multiparameter}.

The SPM operator encodes both the optimal measurement and estimator.
The eigendecomposition 
\begin{equation}
    \label{eq:eigendecomposition_SPM}
    \mathcal{S} = \int \dd{s} \mathcal{P}(s) s,
\end{equation}
then defines the optimal strategy \cite{rubio2024first}: the optimal measurements are projectors $\mathcal{P}(s)$ onto the eigenspace of $\mathcal{S}$, and the outcomes $s$ are post-processed by the estimator $\tilde{\theta}(s)= f^{-1} (s)$.
By virtue of \cref{eq:lyapunov}, the eigenvalue $s$ satisfies \cite{rubio2022quantum}
\begin{equation}
    \label{eq:eigenvalues_SPM_posterior_mean}
    s = \frac{\mathrm{Tr}[\mathcal{P}(s) \bar{\rho}]}{\mathrm{Tr}[\mathcal{P}(s) \rho_0]}= \int \dd{\theta} p(\theta | s) f(\theta),
\end{equation}
from where one recovers the classical Bayes estimator
\begin{equation}
    \label{eq:posterior_mean_estimator}
    \tilde{\theta}(s) = f^{-1} \left[ \int \dd{\theta} p(\theta | s) f(\theta)  \right]. 
\end{equation}
For $f(\theta)=\theta$, \cref{eq:posterior_mean_estimator} reduces to the PM estimator \cite{jaynes2003probability}.  

A key advantage of Bayesian QET is that any measurement $\mathcal{P}(x)$ and estimator $\tilde{\theta}(x)$ pair giving rise to the minimum \cref{eq:MSL_min_global} is \emph{globally} optimal, in the sense that the loss function has been minimised on average taking into account all possible parameter values.
This stands in contrast to the parameter-dependent nature of local QET \cite{demkowicz2015quantum}, notwithstanding the formally analogous role of SPM operators and SLDs as the generators of the optimal strategy.

\section{Constrained optimisation of the MSL}
\label{sec:constrained-optimisation}

\cref{eq:lyapunov} provides a complete solution to the search for optimal estimation strategies under the square loss criterion in \cref{eq:quadratic_loss_f}.
However, since $\mathcal{S}$ here acts on an infinite-dimensional Hilbert space, its spectral measurement must generally be obtained numerically via a Fock-space truncation, whose dimension---and hence computational cost---grows with the required accuracy.
Furthermore, the resulting strategy will depend on the truncation level and does not, in general, admit a closed-form description in terms of standard pre-processing and measurements, making it difficult to assess what experiment to implement.
These challenges motivate a different approach.

Rather than directly optimising the MSL in \cref{eq:average_info_error_projective_POVM} over all measurements and estimators, i.e. over all Hermitian $\mathcal{M}_1$, we restrict $\mathcal{M}_1$ to a finite subspace of Hermitian operators.
As we shall see, this procedure yields a convex quadratic optimisation problem in the expansion coefficients, whose solution is given by a linear system of equations.
As in the full SPM construction, the associated estimation strategy is defined by the spectral decomposition of the resulting constrained analogue of the SPM operator.

\subsection{Constrained optimum}
Let $\mathcal{V}= \operatorname{span}( B_1, \ldots, B_d )$ be a finite-dimensional real subspace of Hermitian operators acting on $\mathcal{H}$, with $\{B_i\}_{i=1}^d$ a fixed Hermitian operator basis.
We restrict the operator $\mathcal{M}_1$ to belong to $\mathcal{V}$, and we write
\begin{equation}
    \label{eq:M_constrained}
    \mathcal{M}_1 = \sum_{i=1}^d \alpha_i B_i, \quad \alpha_i \in \mathbbm{R}.
\end{equation}

Up to an additive constant, minimising the MSL in \cref{eq:average_info_error_projective_POVM} amounts to minimising the functional
\begin{equation}
    \label{eq:J-functional}
    J(\mathcal{M}_{1}) 
    \coloneqq \mathrm{Tr}(\rho_0 \mathcal{M}^2_{1}) - 2\mathrm{Tr}(\bar{\rho} \mathcal{M}_{1}),
\end{equation}
which, for \cref{eq:M_constrained}, becomes
\begin{equation}
    \label{eq:MSL_constrained_alpha}
    \begin{aligned}
    J(\vb*{\alpha}) &= \sum_{i,j} \alpha_i \alpha_j \mathrm{Tr}(\rho_0 B_i B_j) - 2\sum_i \alpha_i \mathrm{Tr}(\bar{\rho} B_i)\\
    & = \vb*{\alpha}^\top G \vb*{\alpha} - 2 \vb{b}^\top \vb*{\alpha},
    \end{aligned}
\end{equation}
where \( \vb*{\alpha} = (\alpha_1, \dots, \alpha_d)^\top \), \(\vb{b} = (b_1, \dots, b_d)^\top\) and
\begin{equation}
    \label{eq:G_and_b_def}
    G_{ij}\coloneq \frac{1}{2}\mathrm{Tr}\,\big(\rho_0\{ B_i,B_j \}\big),
    \quad
    b_i \coloneq \mathrm{Tr}(\bar{\rho} B_i).
\end{equation}
Here, $\{ \cdot, \cdot \}$ denotes the anti-commutator.

Throughout, we assume the traces defining $G$ and \(\boldsymbol{b}\) to be finite, so that the above functional is well-defined.
A natural choice for Gaussian states is to take the $B_i$ to be polynomials in the quadratures $\hat{q}$ and $\hat{p}$.\footnote{
Note that the only requirement on the $B_i$ is that they are Hermitian, since this ensures that $\mathcal{M}_1$ is Hermitian and that $G$ is real and positive semidefinite. For polynomial bases, Weyl (symmetric) ordering is a natural choice, but any Hermitian basis is equally valid.
For example, the operator basis built from the first $n$ Fock states $\{\ket{j}\bra{k} + \ket{k}\bra{j},i(\ket{j}\bra{k}-\ket{k}\bra{j}) \}$, for $ 0 \leq j\leq k \leq n-1$, recovers the Fock-space truncation of $\mathcal{S}$, usually calculated numerically.
\label{footnote:operator_basis}
} 
The moments $\mathrm{Tr}[\rho(\theta) B_i]$ and $\mathrm{Tr}[\rho(\theta) \{B_i,B_j\}]$ are then finite and, by Wick's (Isserlis') theorem, can be computed entirely from the first and second moments $\vb{r}(\theta)$ and $V(\theta)$ \cite{serafini2017quantum,isserlis1918formula}.
Convergence of the prior-weighted integrals then depends on how these moments grow with the encoded parameter $\theta$ and on the tail behaviour of the prior. 
This holds for all examples in this work, but must otherwise be verified on a case-by-case basis.

Since $J(\vb*{\alpha})$ in \cref{eq:MSL_constrained_alpha} is a convex quadratic functional, its stationary points satisfy
\begin{equation}
    \label{eq:S_constrained_opt_solutions}
    G \vb*{\alpha}^\text{opt}=  \vb{b}.
\end{equation}
This follows by taking the gradient $\nabla_{\vb*{\alpha}} J(\vb*{\alpha}) = 2 G \vb*{\alpha} -2\vb{b}$ and imposing $\nabla_{\vb*{\alpha}} J(\vb*{\alpha}^\text{opt})=0$. 
The Hessian is $\nabla^2_{\vb*{\alpha}} J(\vb*{\alpha})=2G$, which implies that $J(\vb*{\alpha}^\text{opt})$ is a global minimum, since $G$ is positive semidefinite by construction. 
Indeed, $\vb*{\alpha}^\top G \vb*{\alpha} = \mathrm{Tr}(\rho_0 \mathcal{M}^2_{1})\geq 0$ for all $\boldsymbol{\alpha}$, with $\rho_0 \succeq 0$ and $\mathcal{M}_{1}=\mathcal{M}_{1}^\dagger$.
Note that \cref{eq:S_constrained_opt_solutions} can also be proven using the explicitly variational approach of \citet{personick1970efficient, personick1971application}, as shown in \cref{appendix:alternate_proof_of_constrained_optimum}.

The uniqueness of this minimum is determined by the structure of $G$.
First, if $G$ is positive definite, it is invertible and \cref{eq:S_constrained_opt_solutions} admits a unique solution $\vb*{\alpha}^\text{opt}= G^{-1} \vb{b}$.
If $G$ is singular, the inverse does not exist, but convexity ensures that the minimum of $J(\vb*{\alpha})$ is still attained, although the minimiser is no longer unique.
In this case, one may select the minimum-norm solution using the Moore--Penrose pseudoinverse.
By introducing the inner product
\begin{equation}
    \label{eq:inner_product_rho0}
    \expval{A,B}_{\rho_0} \coloneq \frac{1}{2}\mathrm{Tr}(\rho_0 \{ A,B\}),
\end{equation}
with respect to which 
\begin{equation}
    \label{eq:gram-norm-version}
    G_{ij}=\expval{B_i,B_j}_{\rho_0}    
\end{equation}
is now a Gram matrix of the basis elements $B_i$, $G$ is found to be singular iff the $B_i$ are linearly dependent with respect to this $\rho_0$-weighted inner product \cite{horn2012matrix}.
A similar inner product also appears naturally in the frequentist analysis of the Helstrom (SLD) and Holevo Cram\'{e}r--Rao bounds \cite{demkowicz2020multiparameter}.

The optimal coefficients $\vb*{\alpha}^\text{opt}$ from \cref{eq:S_constrained_opt_solutions} then define a subspace-optimal operator
\begin{equation}
    \label{eq:SV_constrained}
    \mathcal{S}_\mathcal{V} \coloneq \sum_i \alpha_i^\text{opt} B_i,
\end{equation}
analogous to the SPM operator but not necessarily satisfying \cref{eq:lyapunov}.
The corresponding MSL is
\begin{equation}
    \label{eq:MSL_constrained}
    \begin{aligned}
        \mathcal{L}(\mathcal{S}_\mathcal{V}) &= \lambda - \vb{b}^\top \vb*{\alpha}^\text{opt} \\
        &= \lambda - \vb{b}^\top G^{-1} \vb{b} \geq \mathcal{L}(\mathcal{S}),
    \end{aligned}
\end{equation}
where the inequality trivially stems from the role of \(\mathcal{L}(\mathcal{S})\) as the global minimum. 
The MSL in \cref{eq:MSL_constrained} depends only on the subspace $\mathcal{V}$ and is invariant under linear changes of basis $B_i \mapsto U_{ij} B_j$.\footnote{Indeed, under such a linear transformation, $\vb{b} \mapsto U \vb{b}$ and $G \mapsto U G U^\top$, so $\vb{b}^\top G^{-1} \vb{b}$ is invariant.}
By convexity, any other set of coefficients $\vb*{\alpha}$ will necessarily give a larger MSL, so that $\mathcal{L}(\mathcal{S}_\mathcal{V})$ provides a lower bound on the MSL achievable by any Hermitian operator in $\mathcal{V}$.

Further applying the Gram--Schmidt algorithm to $\{B_i\}$ produces an orthogonal basis $\{\tilde{B}_i\}$ that leads to a diagonal Gram matrix $\tilde{G}$. 
The MSL then separates as
\begin{equation}
    \label{eq:MSL_constrained_diagonal}
    \mathcal{L}(\mathcal{S}_\mathcal{V}) = \lambda - \sum_i \frac{\tilde{b}_i^2}{\tilde{G}_{ii}},
\end{equation}
where $\tilde{b}_i=\textrm{Tr}(\bar{\rho} \tilde{B}_i)$ and $\tilde{G}_{ii}=\textrm{Tr}(\rho_0 \tilde{B}_i^2) $.
Written in this way, each term in \cref{eq:MSL_constrained_diagonal} measures the component of $\bar{\rho}$ along a single orthogonal direction $\tilde{B}_i$---a linear combination of the original $\{B_i\}$---and one observes that a direction is irrelevant for the MSL if $\tilde{b}_i=0$.

Note that this approach is structurally similar to the Ritz--Galerkin method in numerical analysis \cite{brenner2008mathematical,ciarlet2002finite}, where the solution of a differential operator equation is approximated by its orthogonal projection onto a finite-dimensional subspace.
Here, the Lyapunov equation in \cref{eq:lyapunov} plays the role of the operator equation, $\expval{\cdot,\cdot}_{\rho_0}$ defines the energy inner product, and the subspace $\mathcal{V}$ is chosen on physical rather than numerical grounds.

To summarise, we have shown that, once an operator subspace $\mathcal{V}=\operatorname{span}( B_1, \ldots, B_d )$ is specified, \cref{eq:SV_constrained,eq:S_constrained_opt_solutions} determine a constrained analogue of the SPM operator $\mathcal{S}_\mathcal{V}$ containing the estimator and measurement with minimum MSL within this subspace in \cref{eq:MSL_constrained}.
This constitutes our first main result.
The construction only requires that the moments $G_{ij} =  \frac{1}{2}\mathrm{Tr}\,\big(\rho_0\{ B_i,B_j \}\big)$ and $b_i = \mathrm{Tr}(\bar{\rho} B_i)$ in \cref{eq:G_and_b_def} are finite; it is otherwise independent of the dimension of the underlying Hilbert space, applying not only to Gaussian states, but to CV systems with any number of modes as well as to finite-dimensional systems.
The examples in \cref{sec:admissible_measurements} and \cref{sec:1d_displacement,sec:1d_squeezing} specialise to single-mode Gaussian probes for clarity.

Before we proceed with assessing how well this constrained optimum works in practice, we next illustrate the types of admissible estimators and measurements that may arise within our framework.

\subsection{Admissible estimators and measurements}
\label{sec:admissible_measurements}
If $\mathcal{S}_\mathcal{V}$ is self-adjoint, the spectral theorem \cite{reed1972methods,hall2013quantum} guarantees the existence of a unique PVM $\mathcal{P}_\mathcal{V}$ such that
\begin{equation}
    \label{eq:S_V-spectral-decomposition}
    \mathcal{S}_\mathcal{V} = \int \dd{\tau} \mathcal{P}_\mathcal{V}(\tau)\,\tau.
\end{equation}
This spectral decomposition defines a PVM whose outcomes are the spectral values $\tau \in \mathbbm{R}$.
Note that $\tau$, unlike the eigenvalues $s$ in \cref{eq:eigenvalues_SPM_posterior_mean}, need not equal the PM for the measurement $\mathcal{P}_\mathcal{V}$; we return to this point and its consequences in  \cref{sec:improved_MSL_posterior_mean}.
% \footnote{\textcolor{blue}{A related caveat applies to the global SPM operator $\mathcal{S}$.
% Non-negativity of the MSL in \cref{eq:MSL_min_global} gives $\norm{\mathcal{S}}^2_{\rho_0}=\textrm{Tr}(\rho_0 \mathcal{S}^2) \leq \lambda$, so, provided $\lambda \leq \infty$, $\mathcal{S}$ is well defined as an element of the space of operators equipped with the inner product in \cref{eq:inner_product_rho0}.
% This is all that is required for \cref{eq:average_info_error_constrained_difference_equality,eq:MSL_chain_of_inequalities}, which we use only to quantify the excess MSL.
% Whether $\mathcal{S}$ is self-adjoint on a (dense) domain of $\mathcal{H}$, and hence whether the globally optimal PVM in \cref{eq:eigendecomposition_SPM} exists, is a separate question that we do not address here
% }}

To illustrate the structure of \cref{eq:S_V-spectral-decomposition} and the estimation strategy it may define, in this section we consider representative choices of basis for a single-mode system; the generalisation to $N>1$ modes is straightforward.

First, consider the basis elements $B_i \in \{ \mathbbm{1}, \hat{q},\hat{p} \}$. 
Then,
\begin{equation}
    \mathcal{S}_\mathcal{V}= \alpha_0 \mathbbm{1} + \alpha_q \hat{q}+\alpha_p \hat{p} = \alpha_0 \mathbbm{1} + \abs{\vb*{\alpha}} \hat{q}_\phi
\end{equation}
is linear in the quadratures, where we have defined $\hat{q}_\phi = \hat{q} \cos \phi + \hat{p} \sin \phi$, $\abs{\vb*{\alpha}}^2 = \alpha_q^2 + \alpha_p^2$ and $\phi = \arctan{\alpha_p/\alpha_q}$.
The operator $\hat{q}_\phi$ is self-adjoint\footnote{The quadratures $\hat{q}$ and $\hat{p}$ are unbounded and can only be defined on proper dense subdomains of $L^2(\mathbbm{R})$; for unbounded operators, being self-adjoint rather than Hermitian is required to guarantee a spectral decomposition. 
The operator $\hat{q}$ is self-adjoint on its maximal domain $\{\psi \in L^2(\mathbbm{R}) : q\psi(q) \in L^2(\mathbbm{R})\}$ \cite[Prop. 9.30]{hall2013quantum}, and $\hat{p}$ on an analogous Fourier-space domain \cite[Prop. 9.32]{hall2013quantum}. The rotated quadrature $\hat{q}_\phi$ is also self-adjoint since it is unitarily equivalent to $\hat{q}$.} and unitarily equivalent to the position operator.
Its spectral PVM corresponds to projectors onto the (generalised) eigenstates of the rotated quadrature $\ket{q_\phi}$, which realise homodyne detection~\cite{scully1997optics,weedbrook2012gaussian,monras2013phase,adesso2014continuous}.
The cases $\phi=0$ and $\phi=\pi/2$ correspond to homodyne detection of position and momentum, respectively.
By virtue of $\mathcal{M}_1$ defined in \cref{eq:Mk_operator_moments}, the estimator maps each homodyne outcome $q_\phi$ to $\tilde{\theta}(q_\phi) = f^{-1}(\alpha_0 + \abs{\vb*{\alpha}} q_\phi)$.

At second order in the quadratures, the most general form that $\mathcal{S}_\mathcal{V}$ can take is
\begin{equation}
    \label{eq:M_quadratic}
    \mathcal{S}_\mathcal{V} = \alpha_0 \mathbbm{1} + \vb*{\alpha}^\top \hat{\vb{R}} + \frac{1}{2} \hat{\vb{R}}^\top C \hat{\vb{R}},
\end{equation}
where $\vb*{\alpha}=\left( \alpha_q, \alpha_p \right)^\top$ and $C=C^\top$ is a $2\times2$ symmetric matrix.
For this configuration, $\mathcal{S}_\mathcal{V}$ is guaranteed to be self-adjoint.\footnote{For real Weyl-quantised polynomials of degree at most two in the quadratures, $\mathcal{S}_\mathcal{V}$ is self-adjoint on a dense domain of $\mathcal{H}$ \cite[Thm.~2.18 and Sec.~4.3]{folland2016harmonic} \cite[Thms.~352,~357]{de2011symplectic}.}
However the spectral structure, and hence the resulting estimator-measurement pair, depends on the signature of $C$.
If $C \succ 0$, the PVM projects onto (in general displaced and squeezed) Fock states.
Such Fock state measurements are non-Gaussian, but they can be implemented experimentally using photon counting and suitable Gaussian pre-processing  \cite{serafini2017quantum,paschotta2025fock,von2019quantum,barbieri2022optical}.
The corresponding estimator maps each Fock state label $n$ to $\tilde{\theta}(n) = f^{-1}(\tau_n)$, where $\tau_n$ is the corresponding eigenvalue of $\mathcal{S}_\mathcal{V}$.

If $C$ has a mixed signature (e.g., for $\hat{q}\hat{p}+\hat{p}\hat{q}$ or $\hat{p}^2-\hat{q}^2$), the spectrum is continuous and unbounded above and below, and the (generalised) eigenstates are non-normalisable scattering states, making the associated PVM difficult to implement \cite{milburn1994hyperbolic,chiribella2006optimal,oh2019optimal}.
Nevertheless, the PVM can be considered as the limit of a sequence of well-defined PVMs on finite-dimensional Fock-space truncations \cite{pegg1989phase}.  

Higher-order polynomials in a \emph{single} quadrature $\hat{q}_\phi$---leading to homodyne measurements with estimators polynomial in the outcome---are always self-adjoint for any degree, since they act as multiplication operators in the eigenbasis of $\hat{q}_\phi$ \cite[Prop. 9.30]{hall2013quantum}.
This guarantees $\mathcal{S}_\mathcal{V}$ to be self-adjoint for the single-quadrature bases used in the examples.
However, polynomials of degree greater than two that \emph{mix} the two quadratures are not in general self-adjoint; there are well-known pathological cases such as $\hat{p}^2 - \hat{q}^4$ \cite[Thm. 9.41]{hall2013quantum}.
For higher-degree mixed bases, self-adjointness must therefore be verified independently.
Nevertheless, even when the spectral decomposition is not available, the MSL $\mathcal{L}(\mathcal{S}_\mathcal{V})$ in \cref{eq:MSL_constrained} remains a valid lower bound on the MSL achievable by any Hermitian operator in $\mathcal{V}$ whenever $G$ and $\vb{b}$ are finite, although this bound need not itself be attainable.

\subsection{Assessing the performance of the constrained optimum}

While $\mathcal{S}_{\mathcal{V}}$ enables the systematic construction of closed-form expressions for optimal estimation strategies in CV settings, such as for Gaussian states, one generally has $\mathcal{S}_{\mathcal{V}} \neq \mathcal{S}$, so that the inequality in \cref{eq:MSL_constrained} is not saturated.
Given this, in the following we first determine under which conditions $\mathcal{S}_{\mathcal{V}} = \mathcal{S}$, then quantify the excess error incurred by restricting the measurement operator when $\mathcal{S}_{\mathcal{V}} \neq \mathcal{S}$, and finally show that one can further reduce the constrained MSL by refining the estimator to the posterior mean.

\subsubsection{Characterisation of exact polynomial solutions}

A necessary condition for the constrained approach to be exact is that $\mathcal{S}$ is a finite-degree polynomial in the quadratures $\hat{q}$ and $\hat{p}$.
This condition is then sufficient whenever $\mathcal{S}$ belongs to the chosen polynomial subspace $\mathcal{V}$.
The following theorem characterises precisely when this condition holds, and is most naturally stated in phase space.

Let $W_X(\vb{r})$ denote the Weyl symbol of a general Hermitian operator $X$, defined in \cref{eq:Wigner_func_characteristic_func_FT}.
For the prior-averaged operators $\rho_0$ and $\bar{\rho}$, due to linearity of the trace and integrals, these are given by
\begin{subequations}
    \label{eq:Weyl_symbols_rhobar_and_rho0}
    \begin{align}
        W_{\rho_0}(\vb{r}) &= \int \dd{\theta} p(\theta) W_{\rho(\theta)}(\vb{r}), \\
        W_{\bar{\rho}}(\vb{r}) &= \int \dd{\theta} p(\theta) f(\theta) W_{\rho(\theta)}(\vb{r}). 
    \end{align}
\end{subequations}
\begin{theorem}
    \label{theorem:1}
    Let $\rho_0$ be a faithful Gaussian state.
    Then $\mathcal{S}$ is a polynomial of degree $R$ in the quadrature operators if and only if $W_{\bar{\rho}}/W_{\rho_0}$ is a polynomial of degree $R$ in the phase-space coordinates $\vb{r}$.
\end{theorem}
The proof is deferred to \cref{appendix:polynomial_S_proof}.
Here, a faithful $\rho_0$ has $\ker{\rho_0} = \{0\}$.

This theorem recovers a well-known result in local QET.
There, the SLD is defined by a Lyapunov-type equation \cite{braunstein1994statistical,holevo2011probabilistic,demkowicz2020multiparameter},
\begin{equation}
     L\rho + \rho L = 2\partial_\theta \rho,
\end{equation}
which is analogous to \cref{eq:lyapunov} under the identifications $\rho_0 \leftrightarrow \rho$ and $\bar{\rho} \leftrightarrow \partial_\theta\rho$. 
For the Gaussian Wigner function defined in \cref{eq:Gaussian_Wigner_func}, the ratio $W_{\partial_\theta\rho}/W_{\rho(\theta)}$ is always a linear and/or quadratic polynomial in $\vb{r}$,\footnote{This follows because $W_{\partial_\theta\rho}=\partial_\theta W_{\rho(\theta)}$ from linearity of the trace and integrals in \cref{eq:characteristic_function,eq:Wigner_func_characteristic_func_FT}, and the fact that $\partial_\theta \mathrm{log}[W_{\rho(\theta)}]$ is at most quadratic in $\vb{r}$.} and so, for faithful $\rho$, \cref{theorem:1} implies that the SLD is a linear and/or quadratic operator in the quadratures.

Prior treatments \cite{monras2013phase,serafini2017quantum,chang2025multiparameter} verify this by proposing a linear and/or quadratic ansatz and checking it satisfies the Lyapunov equation, and uniqueness of the solution then determines the SLD on the support of $\rho$ \cite{genoni2019non}. 
\cref{theorem:1} provides an alternative, ansatz-free characterisation: the polynomial degree of $\mathcal{S}$  is read off directly from the ratio $W_{\bar{\rho}}/W_{\rho_0}$.
This is particularly useful in the CV Bayesian setting, where $\mathcal{S}$ is generically an infinite series in the quadratures and no polynomial ansatz is apparent \textit{a priori}.
However, when the prior-averaged state $\rho_0$ is non-Gaussian, the hypothesis of \cref{theorem:1} fails, so $\mathcal{S}$ may not admit a finite polynomial representation and $\mathcal{S}_\mathcal{V}$ should be understood as an approximation to $\mathcal{S}$.

\subsubsection{Geometric interpretation of constrained optimality}
\label{sec:geometric_interpretation}

For cases where \(\mathcal{S}_{\mathcal{V}} \neq \mathcal{S}\), the difference \(\mathcal{L}(\mathcal{S}_{\mathcal{V}}) - \mathcal{L}(\mathcal{S})\) admits a simple expression.  
As shown in \cref{appendix:proof_average_info_error_quadratic_inequality}, the MSL associated with an arbitrary Hermitian operator $X$ can always be written as
\begin{equation}
    \label{eq:average_info_error_constrained_difference_equality}
     \mathcal{L}(X) = \mathcal{L}(\mathcal{S}) + \norm{X-\mathcal{S}}^2_{\rho_0},
\end{equation}
where \(\mathcal{L}(X)\) is defined as in \cref{eq:average_info_error_projective_POVM} and the norm with respect to the $\rho_0$-weighted inner product in \cref{eq:inner_product_rho0} as \(\norm{A}_{\rho_0}^2 \coloneq \expval{A, A}_{\rho_0}\).
Taking $X=\mathcal{S}_{\mathcal{V}}$, one has
\begin{equation}
     \label{eq:absolute-MSL-difference}
     \mathcal{L}(\mathcal{S}_{\mathcal{V}}) - \mathcal{L}(\mathcal{S})
     =
     \norm{\mathcal{S}_{\mathcal{V}}-\mathcal{S}}^2_{\rho_0},
\end{equation}
so that \(\mathcal{S}_{\mathcal{V}}\) and \(\mathcal{S}\) achieve the same MSL iff
$\norm{\mathcal{S}_\mathcal{V}-\mathcal{S}}_{\rho_0}=0$, i.e., when the difference
$\mathcal{S}_\mathcal{V}-\mathcal{S}$ vanishes on the support of $\rho_0$.
If $\rho_0$ is faithful, then this condition is equivalent to
$\mathcal{S}\in\mathcal{V}$, i.e., $\mathcal{S}_\mathcal{V}=\mathcal{S}$ as operators.

More importantly, \cref{eq:average_info_error_constrained_difference_equality} leads to a natural geometric interpretation for
\(\mathcal{S}_{\mathcal{V}}\).
Namely, substituting
\cref{eq:average_info_error_constrained_difference_equality}
into the constrained optimisation problem
\begin{equation}
    \mathcal{S}_\mathcal{V}
    =
    \underset{\mathcal{M}_1 \in \mathcal{V}}{\text{argmin}}
    \; \mathcal{L}(\mathcal{M}_1),
\end{equation}
yields
\begin{equation}
    \label{eq:constrained_minimum_as_projection}
    \mathcal{S}_\mathcal{V}
    =
    \underset{\mathcal{M}_1 \in \mathcal{V}}{\text{argmin}}
    \: \norm{\mathcal{M}_1 - \mathcal{S}}^2_{\rho_0}, 
\end{equation}
showing that $\mathcal{S}_\mathcal{V}$ is the orthogonal projection of the global solution $\mathcal{S}$ onto the chosen operator subspace $\mathcal{V}$ with respect to the $\rho_0$-weighted inner product.
For this reason, we henceforth refer to $\mathcal{S}_\mathcal{V}$ as the \emph{projected SPM operator}.
\cref{eq:absolute-MSL-difference,eq:constrained_minimum_as_projection} constitute our second main result.

\subsubsection{Improved estimator based on posterior mean}
\label{sec:improved_MSL_posterior_mean}

So far we have performed a constrained optimisation of the MSL by restricting the operator $\mathcal{M}_1$ to an operator subspace $\mathcal{V}$. 
This optimisation jointly determines both a measurement, via the spectral PVM of $\mathcal{S}_\mathcal{V}$, and an associated estimator obtained from (post-processing) the corresponding spectral values.
By construction, this pair minimises $\mathcal{L}$ within the restricted operator space.

Such an estimator is generally different from the classical Bayes estimator minimising the MSL, given in \cref{eq:posterior_mean_estimator}.
The reason is that the SPM operator admits the spectral decomposition
$\mathcal{S}=\int \dd{s} \mathcal{P}(s)\, s$,
where the eigenvalues $\{s\}$ coincide with the PM values in \cref{eq:eigenvalues_SPM_posterior_mean}.
This follows from the fact that $\mathcal{S}$ satisfies the Lyapunov equation in \cref{eq:lyapunov}.
In contrast, since the projected SPM operator $\mathcal{S}_\mathcal{V}$ does \emph{not} satisfy such a Lyapunov equation, its eigenvalues are not, in general, equal to the PM associated with its spectral measurement.

This observation motivates a natural refinement.
Having determined the spectral PVM $\mathcal{P}_\mathcal{V}(\tau)$ of the projected SPM operator $\mathcal{S}_\mathcal{V}$ from \cref{eq:S_V-spectral-decomposition}, one may replace the constrained estimator with the classical Bayes estimator coming from the likelihood induced by $\mathcal{P}_\mathcal{V}(\tau)$, i.e., \(\tilde{\theta}(\tau)\) written as in \cref{eq:posterior_mean_estimator}, but with the posterior
\begin{equation}
    p(\theta|\tau) \propto p(\theta) \mathrm{Tr}(\mathcal{P}_\mathcal{V}(\tau) \rho(\theta)).
\end{equation}
Since the estimator in \cref{eq:posterior_mean_estimator} minimises the MSL for a fixed measurement, this substitution can only reduce (or leave unchanged) the loss.

The resulting performance therefore satisfies the chain of inequalities
\begin{equation}
    \label{eq:MSL_chain_of_inequalities}
    \mathcal{L}(\mathcal{M}_1)
    \geq
    \mathcal{L}(\mathcal{S}_\mathcal{V})
    \geq
    \mathcal{L}(\tilde{\theta}_{\text{PM}},\mathcal{P}_\mathcal{V})
    \geq
    \mathcal{L}(\mathcal{S}).
\end{equation}
The first inequality holds for any $\mathcal{M}_1 \in \mathcal{V}$, as $\mathcal{S}_\mathcal{V}$ minimises the MSL within the subspace.
The second inequality follows from replacing the estimator obtained by the spectral values of $\mathcal{S}_\mathcal{V}$ with the classical Bayes estimator associated with the same measurement.
The third inequality expresses the global optimality of the SPM operator $\mathcal{S}$.

In practice, the following protocol is enabled:
(i) determine the projected SPM operator $\mathcal{S}_\mathcal{V}$ within a suitably chosen subspace,
(ii) implement its spectral PVM $\mathcal{P}_\mathcal{V}(\tau)$, and
(iii) apply the classical Bayes estimator associated with this measurement in post-processing.
The resulting MSL can be shown \cite{rubio2020bayesian,rubio2022quantum} to be
\begin{equation}
    \label{eq:MSL_general_PM_estimator}
    \mathcal{L}(\tilde{\theta}_{\text{PM}},\mathcal{P}_\mathcal{V})
    = \lambda - \int \dd{\tau}
    \frac{\mathrm{Tr}[\mathcal{P}_\mathcal{V}(\tau)\bar{\rho}]^2}{\mathrm{Tr}[\mathcal{P}_\mathcal{V}(\tau)\rho_0]},
\end{equation}
and the excess error with respect to the global optimum satisfies
\begin{equation}
    \mathcal{L}(\tilde{\theta}_{\text{PM}},\mathcal{P}_\mathcal{V})
    -
    \mathcal{L}(\mathcal{S})
    \leq
    \norm{\mathcal{S}_\mathcal{V} - \mathcal{S}}^2_{\rho_0}.
\end{equation}
This procedure provides the best-performing closed-form estimation strategy within the framework and constitutes our third main result.

\section{Applications}
\label{sec:examples}

We now illustrate the framework through a range of representative examples.
In order to carry out meaningful comparisons, we will use the \emph{relative} MSL
\begin{equation}
    \label{eq:relative_MSL_figure_of_merit}
    \mathcal{L}_R(\tilde{\theta},\mathcal{P}) \coloneq \frac{\mathcal{L}(\tilde{\theta},\mathcal{P}) - \mathcal{L}(\mathcal{S})}{\mathcal{L}(\mathcal{S})} = \frac{\norm{\mathcal{M}_1 - \mathcal{S}}^2_{\rho_0}}{\mathcal{L}(\mathcal{S})},
\end{equation}
for an arbitrary estimator $\tilde{\theta}(z)$ and PVM $\mathcal{P}(z)$ collected in the Hermitian operator $\mathcal{M}_1 = \int \dd{z} \mathcal{P}(z) f[\tilde{\theta}(z)]$.

\subsection{Displacement in 1D}
\label{sec:1d_displacement}

Consider a single mode prepared in a reference state $\rho_\mathrm{in}$ and subject to an unknown displacement $\theta$ along the $q$ quadrature. 
The resulting state is
\begin{equation}
\label{eq:displaced-encoding}
\rho(\theta)=D(\theta) \rho_\mathrm{in} D^{\dagger}(\theta),
\end{equation}
where $D(\theta)=e^{-i \theta \hat{p}}$.
If the input probe $\rho_\mathrm{in}$ is a Gaussian state with mean $\overline{\vb{r}}=(q_0,p_0)^\top$ and covariance matrix $V$, with $V_{qp}=0$ for simplicity, the corresponding Wigner function reads
\begin{equation}
\label{eq:Wigner_func_displacement}
W_\theta(q, p) \propto \exp \left[-\frac{(q - q_{0} - \theta)^{2}}{2 V_{qq}} - \frac{(p - p_{0})^{2}}{2 V_{pp}}\right].
\end{equation}

Since the displacement operator translates the mean along the $q$ direction, $\theta$ acts as a location parameter with $f(\theta)=\theta$, yielding the loss function $\ell(\tilde{\theta}, \theta) = (\tilde{\theta} - \theta)^2$.
One must specify, in addition, a prior. 
We first consider a Gaussian prior with known mean $\mu_0$ and variance $\sigma_0^2$, for which homodyne measurements are known to be optimal \cite{personick1971application,morelli2021bayesian}.

In such a configuration, $\rho_0$ corresponds to a faithful Gaussian state.\footnote{A single-mode Gaussian state $\rho_0$ is faithful iff its covariance matrix satisfies $\det(V_{\rho_0})>1/4$, or, equivalently, the symplectic eigenvalue is $\nu > 1/2$ \cite{adesso2014continuous,serafini2005quantifying}.
Here, $\det(V_{\rho_0}) = V_{pp}(V_{qq} + \sigma_0^2$).
Since the input state satisfies the uncertainty relation $V_{qq}V_{pp} \geq 1/4$, $\det(V_{\rho_0}) > 1/4$ for any $\sigma_0^2 >0$.}
Calculating the ratio $W_{\bar{\rho}}(\vb{r})/W_{\rho_0}(\vb{r})$ defined in \cref{eq:Weyl_symbols_rhobar_and_rho0} renders
\begin{equation}
    \frac{W_{\bar{\rho}}(\vb{r})}{W_{\rho_0}(\vb{r})}
    = \frac{\sigma_0^2 (q - q_0) + \mu_0 V_{qq}}{\sigma_0^2 + V_{qq}}.
\end{equation}
Consequently, \cref{theorem:1} guarantees that the full SPM operator $\mathcal{S}$ is exactly linear in the quadratures for any prior width.
We can thus choose the linear subspace $\mathcal{V}$ spanned by $B_i \in \{\mathbbm{1}, \hat{q}\}$, so that the projection of the SPM operator is trivial and we have $\mathcal{S}_\mathcal{V} = \mathcal{S}$ exactly.

To calculate $\mathcal{S}_\mathcal{V}$, we define, for convenience, the centred operator $\Delta \hat{q} \coloneq \hat{q} - \langle \hat{q} \rangle_{\rho_0} \mathbbm{1}$, where, using \cref{eq:prior_averaged_state_moments},
\begin{equation}
    \langle \hat{q} \rangle_{\rho_0}
    = \int \dd{\theta} p (\theta) \mathrm{Tr}[\rho(\theta) \hat{q}]
    = q_0 + \mu_0.
\end{equation}
In this centred basis, the Gram matrix $G$ in \cref{eq:gram-norm-version} is diagonal, so that the linear system in \cref{eq:S_constrained_opt_solutions} reads
\begin{equation}
    \label{eq:S_constrained_opt_solutions_linear_displacement}
    \mqty( 1 & 0 \\ 0 & V_{qq} + \sigma_0^2) \mqty(\alpha_0 \\ \alpha_q) = \mqty(\mu_0 \\ \sigma_0^2),
\end{equation}
where $G_{qq} = \mathrm{Tr}[\rho_0 (\Delta \hat{q})^2] = V_{qq}+\sigma_0^2$, $b_1 = \mathrm{Tr}(\bar{\rho} \, \mathbbm{1}) = \mu_0 $, and $b_{q} = \mathrm{Tr}(\bar{\rho} \: \Delta \hat{q}) = \sigma_0^2$ have been used.
Solving for the coefficients \(\vb*{\alpha}\) yields the projected SPM operator
\begin{equation}
    \label{eq:S_operator_displacement_analytic}
    \begin{aligned}
        \mathcal{S}_\mathcal{V} &= \alpha_0 \mathbbm{1} + \alpha_q  \Delta \hat{q} \\
        &= \mu_0 \mathbbm{1} + \frac{\sigma_0^2}{\sigma_0^2 + V_{qq}} \left[ \hat{q} - (q_0 + \mu_0)\mathbbm{1} \right] \\
        &= \frac{\sigma_0^2 (\hat{q} - q_0\mathbbm{1}) + \mu_0 V_{qq} \mathbbm{1}}{\sigma_0^2 + V_{qq}},
    \end{aligned}
\end{equation}
which, as discussed, corresponds to the exact SPM operator $\mathcal{S}$ in this example. 

We next calculate the optimal strategy.
Given the resolution of the identity $\int \dd{q}\, \ket{q}\bra{q} = \mathbbm{1}$, the spectral decomposition of $\mathcal{S}$ is given by
\begin{equation}
    \mathcal{S}= \int \dd{q}\ketbra{q} (\alpha_0+\alpha_q \Delta q),
\end{equation}
where $\Delta q \coloneq q- \langle \hat{q} \rangle_{\rho_0}$.
The measurements are given by the spectral PVM $\mathcal{P}(q)\dd{q}=\ketbra{q} \dd{q}$, which is homodyne detection of the $q$ quadrature.
The outcome $q$ is post-processed by the estimator $\tilde{\theta}(q)=f^{-1}(\alpha_0+\alpha_q \Delta q)$, which for an identity $f$ yields
\begin{equation}
    \label{eq:x_homodyne_linear_estimator}
    \tilde{\theta}(q)= \frac{\sigma_0^2 ( q - q_0) + \mu_0 V_{qq}}{\sigma_0^2 + V_{qq}}.
\end{equation}
This is a linear \emph{shrinkage} estimator: a convex combination of the measurement outcome and \emph{a priori} mean which interpolates between prior-dominated $\sigma_0^2 \to 0$ and measurement-dominated $\sigma_0^2 \to \infty$ regimes \cite{lehmann2006theory}.
Since $\mathcal{S}_\mathcal{V} = \mathcal{S}$, Eq.~\eqref{eq:x_homodyne_linear_estimator} also coincides with the PM estimator for homodyne detection.

Finally, the MSL of this strategy can be calculated using \cref{eq:MSL_constrained}, which turns out to be
\begin{equation}
    \label{eq:MSL_displacement}
    \mathcal{L}(\mathcal{S}) = \left( \frac{1}{V_{qq}} + \frac{1}{\sigma_0^2} \right)^{-1}.
\end{equation}
This recovers the known result \cite{morelli2021bayesian}.

While a Gaussian prior is a natural choice in the context of Gaussian states and often leads to tractable results \cite{morelli2021bayesian}, it is not the only meaningful representation of prior ignorance in this context \cite{kass1996the, boeyens2025role}.
Moreover, the framework developed in \cref{sec:constrained-optimisation} is not limited to Gaussian priors. 
For that reason, we now consider a uniform prior
\begin{equation}
    \label{eq:uniform-prior}
    p(\theta) = \frac{1}{\Delta}, \quad \theta \in \left[\mu_0 - \frac{\Delta}{2},\mu_0 + \frac{\Delta}{2}\right],
\end{equation}
with mean $\mu_0$ and variance $\sigma_0^2 = \Delta^2/12$.
For this prior, which represents \emph{maximum} ignorance for a location parameter \cite{jaynes2003probability}, $\rho_0$ is a mixture of displaced Gaussian states that is no longer Gaussian.
Consequently, \cref{theorem:1} no longer guarantees that the SPM operator is a finite polynomial in the quadratures, and its projection onto the subspace $\mathcal{V}$ should therefore be viewed as an approximation to the full $\mathcal{S}$.
In the following, we consider several choices of operator subspaces and compare their performance in \cref{fig:displacement_estimation} using the relative MSL in \cref{eq:relative_MSL_figure_of_merit}.

\begin{figure}
    \centering
    \includegraphics[width= 0.99\linewidth]{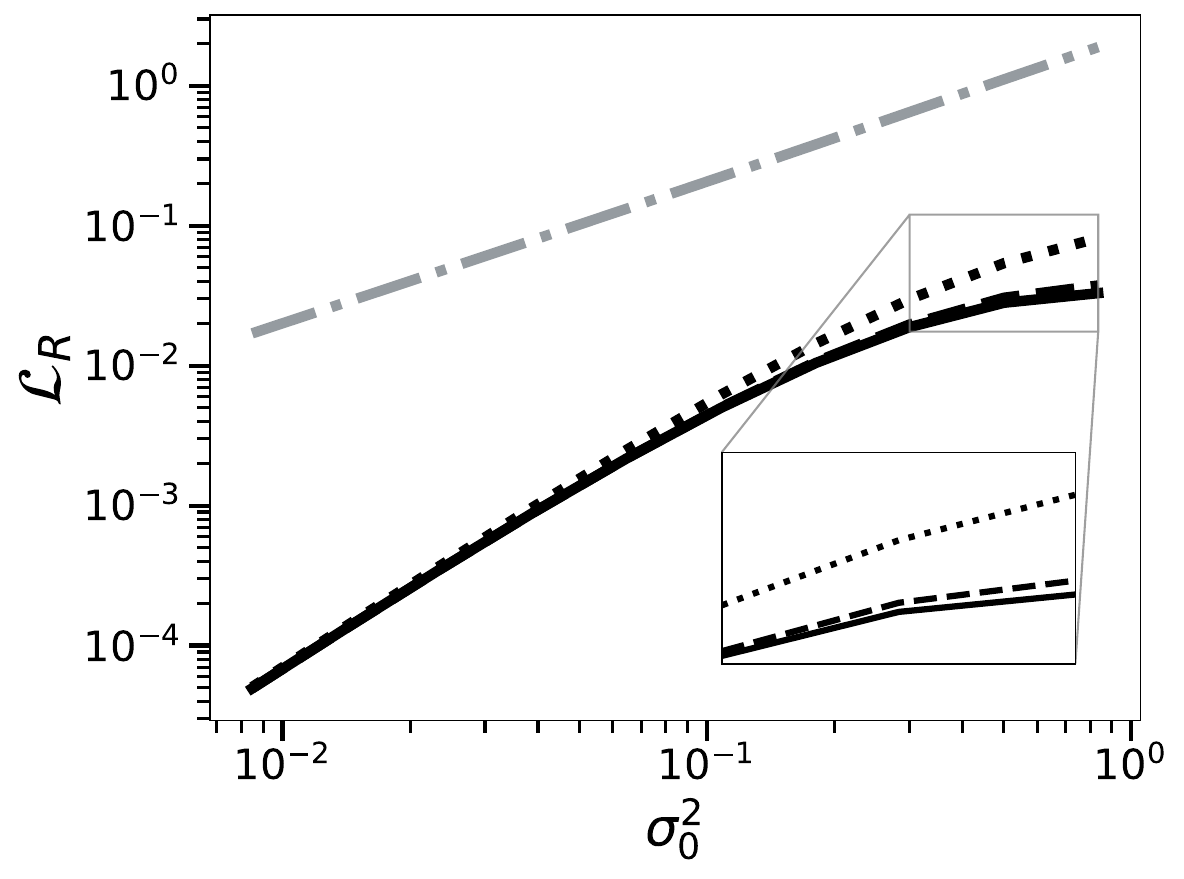}
    \caption{Relative MSL in \cref{eq:relative_MSL_figure_of_merit} as a function of the prior variance $\sigma_0^2$ for estimating a displacement parameter $\theta$ as per \cref{eq:displaced-encoding}.
    Here, we use a uniform prior with width $\Delta$ and variance $\sigma_0^2 = \Delta^2/12$, and $\rho_\textrm{in}$ is a coherent state with $\alpha=0.5(1+i)$.
    The dotted and dashed curves denote the relative MSL of the linear $\{ \mathbbm{1}, \Delta \hat{q}\}$ and cubic $\{ \mathbbm{1}, \Delta \hat{q},\Delta \hat{q}^3\}$ subspaces, respectively, representing $q$-homodyne measurements with linear and cubic estimators.
    The solid curve shows the performance of $q$-homodyne measurements combined with the PM estimator, which includes contributions from all orders.
    The grey dot-dashed curve indicates the prior MSL ratio for reference.
    Numerical values of $\mathcal{L}(\tilde{\theta}_{PM},\mathcal{P}_\mathcal{V})$ and $\mathcal{L}(\mathcal{S})$ are computed using a Fock-space truncation of 40.
    }
    \label{fig:displacement_estimation}
\end{figure}

The simplest choice is the linear basis $B_i \in \{ \mathbbm{1}, \Delta \hat{q} \}$.
One finds that the projected SPM operator takes the same form as \cref{eq:S_operator_displacement_analytic}, with $\sigma_0^2 = \Delta^2/12$.
However, since $\mathcal{S}_\mathcal{V} \neq \mathcal{S}$, the associated linear estimator $\tilde{\theta}(q) = \alpha_0 + \alpha_q \Delta q$ no longer coincides with the PM estimator.
As shown in \cref{fig:displacement_estimation}, where this \(q\)-homodyne measurement with the uniform prior in \cref{eq:uniform-prior} is implemented, the resulting linear estimator (dotted line) is nevertheless close to optimal, as evidenced by the relative MSL \(\mathcal{L}_R\), which is already below \(10^{-1}\) for the largest prior width \(\sigma_0\) considered and decreases rapidly thereafter.

It is natural to ask whether including higher-order terms $\{\mathbbm{1}, \hat{q},\hat{q}^2, \ldots \}$, amounting to \(q\)-homodyne measurements with polynomial estimators of higher degree, leads to an improvement in the MSL. 
To answer this question, we first orthogonalise the basis using the Gram--Schmidt procedure (cf. \cref{sec:constrained-optimisation}) such that the Gram matrix $G$ is diagonal.
Then the MSL separates as in \cref{eq:MSL_constrained_diagonal} and each direction contributes independently.
Starting from $\{\mathbbm{1}, \hat{q},\hat{q}^2 \}$, we obtain the orthogonal set
$\{\mathbbm{1}, \Delta\hat{q},\Delta\hat{q}^2 - \langle \Delta\hat{q}^2 \rangle_{\rho_0}\mathbbm{1}\}$, where we used the centred operator $\Delta \hat{q} = \hat{q} -\langle \hat{q} \rangle_{\rho_0}\mathbbm{1}$.
In this basis, we find that
\begin{equation}
    \tilde{b}_{q^2}=\mathrm{Tr}[\bar{\rho} (\Delta\hat{q}^2 - \langle \Delta\hat{q}^2 \rangle_{\rho_0}\mathbbm{1})]=0,
\end{equation}
so that this orthogonalised basis element does not contribute to the MSL.
Since $\Delta \hat{q}^2$ differs from it only by a multiple of the identity, which the subspace already contains, enlarging $\mathcal{V}$ from $\operatorname{span}(\mathbbm{1},\Delta \hat{q})$ to $\operatorname{span}(\mathbbm{1},\Delta \hat{q},\Delta \hat{q}^2)$ leaves the MSL unchanged.
In fact, this holds at all even orders: for a prior symmetric about $\mu_0$, one finds that $\mathrm{Tr}[\bar\rho\,(\Delta\hat{q}^{2n}-\langle \Delta\hat{q}^{2n}\rangle_{\rho_0} \mathbbm{1})]=0$, so that the centred element $\Delta\hat{q}^{2n}$ does not contribute to the MSL.\footnote{This follows by substituting $\theta\to2\mu_0-\theta$ and
$q\to2(q_0+\mu_0)-q$, which leaves $p(\theta)$, $p(q|\theta)$ and
$\Delta q^{2n}$ invariant.
The same symmetry gives $\langle \Delta \hat{q}^{2n},\Delta \hat{q}^{2m+1} \rangle_{\rho_0} = \langle\Delta \hat{q}^{2(n+m)+1} \rangle_{\rho_0}=0$, so the Gram system \cref{eq:S_constrained_opt_solutions} decouples into even and odd blocks. 
In the even block, only the identity survives with $\alpha_0 = \textrm{Tr}(\bar{\rho} \mathbbm{1}) = \mu_0$.}
Equivalently, the PM is an odd function of $\Delta q$, and so
even-order polynomial corrections do not contribute.

This observation motivates the basis
% \footnote{\textcolor{blue}{Note that the centering matters once the quadratic direction is dropped: $\Delta\hat{q}^3$ has no component along the non-contributing direction $\Delta\hat{q}^2 - \langle \Delta\hat{q}^2 \rangle_{\rho_0}\mathbbm{1}$, whereas $\hat{q}^3$ has one proportional to $\langle \hat{q} \rangle_{\rho_0}$. Within $\operatorname{span}(\{\mathbbm{1}, \hat{q},\hat{q}^3 \})$, this component cannot be removed independently of the cubic coefficient.}}
$\{\mathbbm{1}, \Delta\hat{q},\Delta\hat{q}^3\}$,
corresponding to \(q\)-homodyne detection with a cubic estimator
$\tilde{\theta}(q) = \alpha_0 + \alpha_q \Delta q + \alpha_{q^3} \Delta q^3$. 
While analytic expressions for the coefficients can be obtained, we omit them here for brevity.
The resulting MSL, shown in \cref{fig:displacement_estimation} (dashed curve), lies between that of the previous linear estimator (dotted curve) and that of the exact PM estimator (\cref{eq:posterior_mean_estimator}, solid curve), and is substantially closer to the latter.

It is not surprising that the PM estimator achieves a lower MSL for the same PVM, since it is not restricted to a finite polynomial truncation---generally,
$\mathcal{L}(\tilde{\theta}_\mathcal{V}, \mathcal{P}_\mathcal{V}) \geq \mathcal{L}(\tilde{\theta}_{\text{PM}}, \mathcal{P}_\mathcal{V})$.
Note that even the combination of $q$-homodyne detection with the PM estimator is not globally optimal here, as evidenced by $\mathcal{L}_R \neq 0$ for the solid curve.
In practice, evaluating the PM typically requires numerical integration to compute its normalisation and moments.
In contrast, the approach in this work provides a systematic and computationally efficient approximation scheme for deriving analytic expressions for the optimal estimator and measurement at each truncation order.

\subsection{Estimation of squeezing}
\label{sec:1d_squeezing}

For our second example, consider the estimation of an unknown squeezing parameter $\theta \in \mathbbm{R}$ acting on a single-mode state $\rho_\mathrm{in}$.
The encoded state is
\begin{equation}
    \label{eq:squeezing-transformation}
    \rho(\theta)=S(\theta) \rho_\mathrm{in} S^{\dagger}(\theta),
\end{equation}
where the single-mode squeezing operator 
\begin{equation}
    \label{eq:squeezing_operator_zero_angle}
    S(\theta)=\exp(i\frac{\theta}{2}\{ \hat{q},\hat{p}\})
\end{equation}
is chosen at angle $\varphi=0$ without loss of generality.\footnote{The single-mode squeezing operator at an arbitrary angle $\varphi$ has the form 
\[
S(\theta,\varphi) = \exp(\frac{\eta^* \hat{a}^2 - \eta \hat{a}{^\dagger}^2}{2}),
\]
where $\eta = \theta e^{i \varphi} \in \mathbbm{C}$~\cite{walls2008quantum,barnett2002methods}.
Defining $R(\varphi)=e^{-i \varphi \hat{a}^\dagger \hat{a}}$, one has $R(\varphi) \hat{a}^2 R^\dagger(\varphi)=e^{2i \varphi}\hat{a}^2$, so that \(S(\theta,\varphi) = R^\dagger(\varphi/2) S(\theta,0) R(\varphi/2)\).
Therefore, estimation at an angle $\varphi$ with probe $\rho_\mathrm{in}$ is unitarily equivalent to estimation at $\varphi=0$ with a rotated probe $\rho_\mathrm{in}^\prime = R(\varphi/2) \rho_\mathrm{in} R^\dagger(\varphi/2)$, where $\varphi$ is assumed to be known.
}
Under this transformation, the mean and covariance matrix of a Gaussian state transform as $\overline{\vb{r}} \mapsto \Sigma \, \overline{\vb{r}}$ and $V \mapsto \Sigma V \Sigma^\top$, where $\Sigma(\theta) = \operatorname{diag}(e^{-\theta},e^{\theta})$ \cite{adesso2014continuous,serafini2017quantum,weedbrook2012gaussian}.
For displaced probes, $\theta$ is encoded in both first and second moments.
Here, we focus instead on the undisplaced case in order to showcase the framework more clearly. 
In this case, the squeezing parameter $\theta$ is encoded entirely in the covariance matrix as
\begin{equation}
    \label{eq:squeezing_covariances}
    \langle\hat{q}^2\rangle_{\rho(\theta)}=V_{qq} e^{-2\theta},
    \quad 
    \langle\hat{p}^2\rangle_{\rho(\theta)} = V_{pp} e^{2\theta},
\end{equation}
while all first moments vanish identically and $ \langle{\frac{1}{2}\{\hat{q},\hat{p}\}}\rangle_{\rho(\theta)}=V_{qp}$ is $\theta$-independent.

One immediate consequence is that 
\begin{equation}
    \label{eq:b_vanish_squeezing_linear_q}
    \mathrm{Tr}[\bar{\rho} \, \hat{q}_\phi^{2n+1}]=\int \dd{\theta}  f(\theta) \, p(\theta) \mathrm{Tr}[\rho(\theta) \hat{q}_\phi^{2n+1} ] = 0
\end{equation}
for any quadrature angle $\phi$ and any $n \geq 0$, since all odd moments of a zero-mean Gaussian $\rho(\theta)$ vanish.
Homodyne detection with any $f$-transformed estimator odd in the outcome $q_\phi$ therefore does not contribute to the MSL in \cref{eq:MSL_constrained}.
The same conclusion follows from the likelihood: $p(q_\phi|\theta)$ is a zero-mean Gaussian in $q_\phi$ and depends on the outcome only through $q_\phi^2$, hence the PM estimator \cref{eq:posterior_mean_estimator} is an even function of $q_\phi$ for any prior.

To calculate the projected SPM operators $\mathcal{S}_\mathcal{V}$ and the associated MSLs for the state in \cref{eq:squeezing-transformation}, we need to specify a prior and a loss function. 
We choose the same Gaussian prior as in \cref{sec:1d_displacement}, for which now the prior-averaged state $\rho_0$ is non-Gaussian.
Consequently, the SPM operator is generally not a finite polynomial in the quadratures so $\mathcal{S}_\mathcal{V} \neq \mathcal{S}$.

As for the loss function, we first note from \cref{eq:squeezing_operator_zero_angle} that $S(\theta_1)S(\theta_2)= S(\theta_1 + \theta_2)$, so that the squeezing parameter $\theta$ acts additively under composition \cite{agarwal2012quantum}.
It is therefore a location parameter with $f(\theta) = \theta$.\footnote{
Specifically, the set $\{ S(\theta) | \: \theta \in \mathbbm{R}\}$ forms a one-parameter group under operator multiplication and is isomorphic to the group of real numbers under addition $(\mathbbm{R},+)$.
The unique (up to normalisation) Haar measure on $(\mathbbm{R},+)$ is the Lebesgue measure $\dd{\theta}$, which leads to $f(\theta) = \theta$ \cite{boeyens2025role,kass1996the}.
Alternatively, one could treat $z = e^{\theta} \in \mathbbm{R}^+$ as a scale parameter, since $S(z_1)S(z_2)= S(z_1 z_2)$, with a scale-invariant loss function $\ell(\tilde{z},z) = \log^2(\tilde{z}/z)$ and $f(z) = \log z$ \cite{rubio2022quantum}.
The only change in the estimation problem is in $\bar{\rho} = \int \dd{z} p_z(z) \rho(z) \log(z)$.
The substitution $z = e^{\theta}$ then recovers the same $\bar{\rho}$ as in the displacement case with $f(\theta)=\theta$ and the transformed prior $p_z(z) =p_\theta(\log{z})/z$.
In particular, a Gaussian $p_\theta$ corresponds to a lognormal $p_z$.
} 
Accordingly, the standard quadratic loss $\ell(\tilde{\theta}, \theta) = (\tilde{\theta} - \theta)^2$ reemerges as in \cref{sec:1d_displacement}, displaying translational invariance \cite{rubio2024first}.

Given this prescription, we first consider the subspace $\mathcal{V}_\text{quad} =
\operatorname{span}(\mathbbm{1},\hat{q},\hat{p},\hat{q}^2,\hat{p}^2,\tfrac{1}{2}\{\hat{q},\hat{p}\})$
of all Hermitian polynomials of degree at most two in the quadratures, introduced in \cref{sec:admissible_measurements}.
It is convenient to organise $\mathcal{V}_\text{quad}$ using the rescaled operators 
\begin{subequations}
    \begin{align}
        \label{eq:squeezing_rescaled_operators_Q} Q_\pm &= \frac{\hat{q}^2 - \tilde{V}_{qq}\mathbbm{1}}{\tilde{V}_{qq}} \; \pm \; \frac{\hat{p}^2- \tilde{V}_{pp}\mathbbm{1}}{\tilde{V}_{pp}}, \\
        \label{eq:squeezing_rescaled_operator_J} \Delta K &= \tfrac12\{\hat q,\hat p\} - \tilde{V}_{qp}\mathbbm{1},
    \end{align}
\end{subequations}
so that $\mathcal{V}_\text{quad}=\operatorname{span}(\mathbbm{1},\hat{q},\hat{p},Q_-, Q_+, \Delta K)$.
Here, we have defined the prior-averaged components of the covariance matrix \cref{eq:covariance_matrix},
\begin{equation}
    \begin{aligned}
        \tilde{V}_{qq} & \coloneq \langle\hat{q}^2\rangle_{\rho_0} = e^{-2\mu_0+2\sigma_0^2} V_{qq}, \\
        \tilde{V}_{pp} &\coloneq \langle\hat{p}^2\rangle_{\rho_0}=e^{2\mu_0 + 2\sigma_0^2} V_{pp}, \\
        \tilde{V}_{qp} &\coloneq \langle \tfrac{1}{2}\{\hat{q},\hat{p}\}\rangle_{\rho_0} = V_{qp},
    \end{aligned}
\end{equation}
where the averages have been evaluated for a Gaussian prior.

In this basis, $Q_-$ is $\rho_0$-orthogonal, with respect to the inner product in \cref{eq:inner_product_rho0}, to every other element.
The identity is likewise orthogonal to the remaining four, since each has zero $\rho_0$-mean by construction.
Moreover, the $\boldsymbol{b}$ vector vanishes along all four of $\hat{q}, \hat{p}, Q_+$, and $\Delta K$.
\footnote{\label{footnote:quadratic_squeezing}Writing \cref{eq:squeezing_rescaled_operators_Q} as $Q_\pm = Q \pm P$, one finds that $\langle Q_- ,Q_+\rangle_{\rho_0}=\|Q\|^2_{\rho_0}-\|P\|^2_{\rho_0}=0$ and $\langle Q_-, \Delta K \rangle_{\rho_0}=0$, since $\langle Q, \Delta K\rangle_{\rho_0}=\langle P, \Delta K\rangle_{\rho_0}=2\tilde{V}_{qp}$.
The linear directions $\hat{q}$ and $\hat{p}$ have $b=0$ and are $\rho_0$-orthogonal to the quadratic elements since all odd moments of a zero-mean Gaussian $\rho(\theta)$ vanish.
As for the other components of the $\vb{b}$ vector, $b_Q=-b_P=-2 \sigma_0^2$, so that $b_{Q_-}=b_Q-b_P=-4\sigma_0^2$ and $b_{Q_+}=b_Q+b_P=0$.
Lastly, the centred $b_{\Delta K}$ vanishes because $\langle \tfrac12\{\hat q,\hat p\} \rangle_{\rho(\theta)}=V_{qp}$ is $\theta$-independent.
Both $b_Q=-b_P$ and $\|Q\|^2_{\rho_0}=\|P\|^2_{\rho_0}$ hold for any prior symmetric around $\mu_0$; under the reflection $\theta \mapsto 2\mu_0 - \theta$, the moments of $\hat{q}$ and $\hat{p}$ end up exchanging roles in the prior-weighted integrals.
}
The linear system in \cref{eq:S_constrained_opt_solutions} therefore decouples into a diagonal block on $\{\mathbbm{1},Q_-\}$ and a second block on $\{\hat{q},\hat{p}, Q_+, \Delta K\}$, which, having a vanishing $\boldsymbol{b}$ vector, contributes nothing to the MSL in \cref{eq:MSL_constrained} and need not be inverted.

The coefficients of the $\{\mathbbm{1},Q_-\}$ basis are then obtained from the linear system
\begin{equation}
    \mqty(1 & 0 \\ 0 & \|Q_-\|^2_{\rho_0} ) \mqty(\alpha_0 \\ \alpha_{Q_-}) = \mqty(\mu_0 \\ b_{Q_-}),
\end{equation}
% In the $\{\mathbbm{1},Q_-, Q_+, \Delta K\}$ basis, \cref{eq:S_constrained_opt_solutions} reduces to
% \begin{widetext}
% \begin{equation}
%     \mqty(1 & 0 & 0 & 0 \\ 0 & \|Q_-\|^2_{\rho_0} & 0 & 0 \\ 0 & 0 & \|Q_+\|^2_{\rho_0} & \langle  Q_+, \Delta K \rangle_{\rho_0} \\ 0 & 0 & \langle Q_+, \Delta K\rangle_{\rho_0} & \|\Delta K\|^2_{\rho_0}) \mqty(\alpha_0 \\ \alpha_{Q_-} \\ \alpha_{Q_+} \\ \alpha_{\Delta K}) = \mqty(\mu_0 \\ b_{Q_-} \\ 0 \\ 0).
% \end{equation}
% \end{widetext}
and solving yields the projected SPM operator and corresponding MSL,
\begin{subequations}
    \begin{align}
        \label{eq:constrained_SPM_squeezing_quadratic}
        \mathcal{S}_{\mathcal{V}_\text{quad}} &= \mu_0 \mathbbm{1} + \frac{b_{Q_-}}{\|Q_-\|^2_{\rho_0}}Q_-, \\
        \label{eq:constrained_MSL_squeezing_quadratic}
        \mathcal{L}(\mathcal{S}_{\mathcal{V}_\text{quad}}) &= \lambda - \mu_0^2 - \frac{b_{Q_-}^2}{\|Q_-\|^2_{\rho_0}}.
    \end{align}
\end{subequations}
Here, $\lambda - \mu_0^2 = \sigma_0^2$, $b_{Q_-} = -4\sigma_0^2$, and
\begin{equation}
    \|Q_-\|^2_{\rho_0} = \frac{1 + 2(3e^{8\sigma_0^2}-1)V_{qq}V_{pp} - 4V_{qp}^2}{e^{4\sigma_0^2}V_{qq}V_{pp}}.
\end{equation}
Note that $\operatorname{span}(\mathbbm{1}, \hat{q}^2,\hat{p}^2)$ already contains the operator $Q_-$ and hence $\mathcal{S}_{\mathcal{V}_\text{quad}}$.
Since $\operatorname{span}(\mathbbm{1},Q_-) \subseteq \operatorname{span}(\mathbbm{1},\hat{q}^2,\hat{p}^2) \subseteq \mathcal{V}_\text{quad}$, all three subspaces attain the same minimum MSL, and the optimisation may equivalently be carried out in the smaller one.

While the above specialises to a Gaussian prior, $\Delta K$ remains uninformative for any prior since $\langle \tfrac{1}{2}\{\hat{q},\hat{p}\}\rangle_{\rho(\theta)}$ is $\theta$-independent.
By contrast, $b_{Q_+}$ vanishes for symmetric priors because $\langle Q_+ \rangle_{\rho(\theta)}$ is even about $\mu_0$.
For asymmetric priors, $b_{Q_+}$ and $\langle  Q_-, Q_+ \rangle_{\rho_0}$ are non-zero in general, so the full $4 \times 4$ linear system must be solved.
Its entries can still be computed explicitly since they only involve prior averages of $e^{\pm 2\theta}$, $e^{\pm 4 \theta}$ and $\theta e^{\pm 2\theta}$, which are obtainable from the moment generating function $M(t) = \langle e^{t \theta} \rangle$ and its first derivative.

% Solving \cref{eq:S_constrained_opt_solutions} yields the projected SPM operator
% \begin{equation}
%     \mathcal{S}_{\mathcal{V}_2} = \mu_0 \mathbbm{1} 
%     +\,\frac{4\sigma_0^2 e^{2\sigma_0^2}(V_{qq} e^{-2\mu_0} \Delta (\hat{p}^2) -V_{pp} e^{2\mu_0} \Delta (\hat{q}^2))}{1+2(3e^{8 \sigma_0^2}-1) V_{qq} V_{pp}-4V_{qp}^2},
% \end{equation}
% with corresponding MSL
% \begin{equation}
%     \mathcal{L}(\mathcal{S}_{\mathcal{V}_2})=\sigma_0^2 -  \frac{16 \sigma_0^4 e^{4\sigma_0^2} V_{qq} V_{pp}}{1 + 2 (3 e^{8 \sigma_0^2} -1) V_{qq}V_{pp}-4V_{qp}^2}.
% \end{equation}
% Since $\mathcal{S}_{\mathcal{V}_\text{quad}} \in \operatorname{span}(\mathbbm{1},Q_-) \subset \mathcal{V}_2 \subset \mathcal{V}_\text{quad}$, with $\mathcal{V}_2 = \operatorname{span}(\mathbbm{1}, \hat{q}^2,\hat{p}^2)$, the two minimum MSL coincide $\mathcal{L}(\mathcal{S}_{\mathcal{V}_2}) =\mathcal{L}(\mathcal{S}_{\mathcal{V}_\text{quad}})$, and it suffices to work with $\mathcal{V}_2$ in what follows.

We evaluate the framework for two initial probes \(\rho_{\mathrm{in}}\): the vacuum and a thermal state with mean photon number $\bar{n}=0.1$, with covariances $V_{qq}=V_{pp}=\bar{n}+\frac{1}{2}$ and $V_{qp}=0$.
For each probe, we compare the relative MSL in \cref{eq:relative_MSL_figure_of_merit} for the constrained estimator induced by the eigenvalues of $\mathcal{S}_\mathcal{V}$ [\cref{eq:constrained_SPM_squeezing_quadratic}] and the PM estimator [\cref{eq:posterior_mean_estimator}], both obtained from the same measurement. 
The results are shown in \cref{fig:squeezing_estimation_zero_mean} (top: vacuum; bottom: thermal) as dashed [\cref{eq:constrained_MSL_squeezing_quadratic}] and solid [\cref{eq:MSL_general_PM_estimator}] green curves, respectively.
Both strategies converge rapidly to the global minimum $\mathcal{L}(\mathcal{S})$, with the thermal probe converging more slowly than the vacuum.
Furthermore, the PM estimator improves upon the constrained estimator across the full range of prior widths, consistent with \cref{eq:MSL_chain_of_inequalities}, with the improvement becoming more pronounced for broader priors.
The inset shows the corresponding absolute MSLs, where the convergence towards the global optimum $\mathcal{L}(\mathcal{S})$ (solid red) can be seen explicitly.

\begin{figure}
    \centering
    \includegraphics[width= 0.99\linewidth]{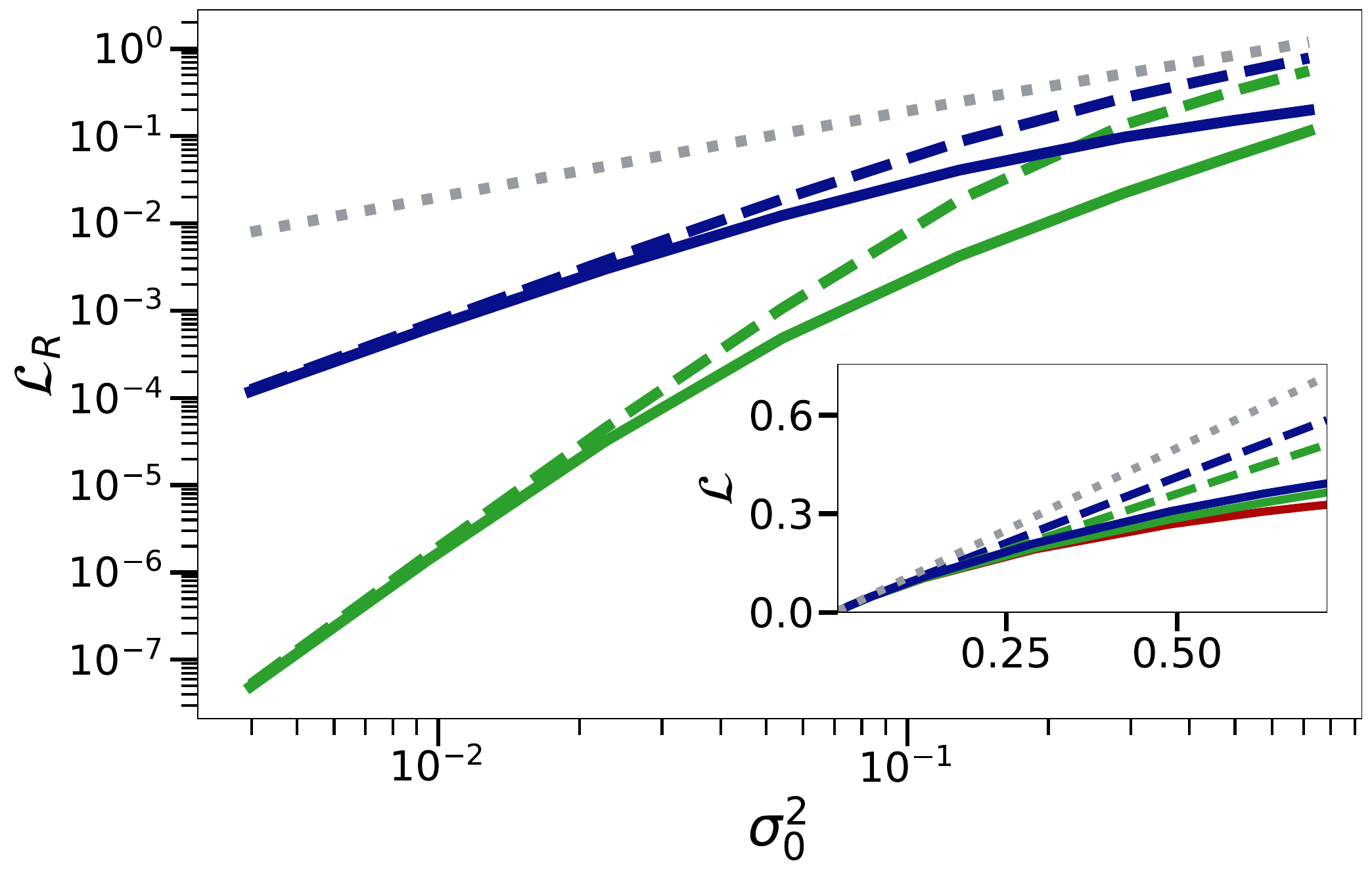}
    \includegraphics[width= 0.99\linewidth]{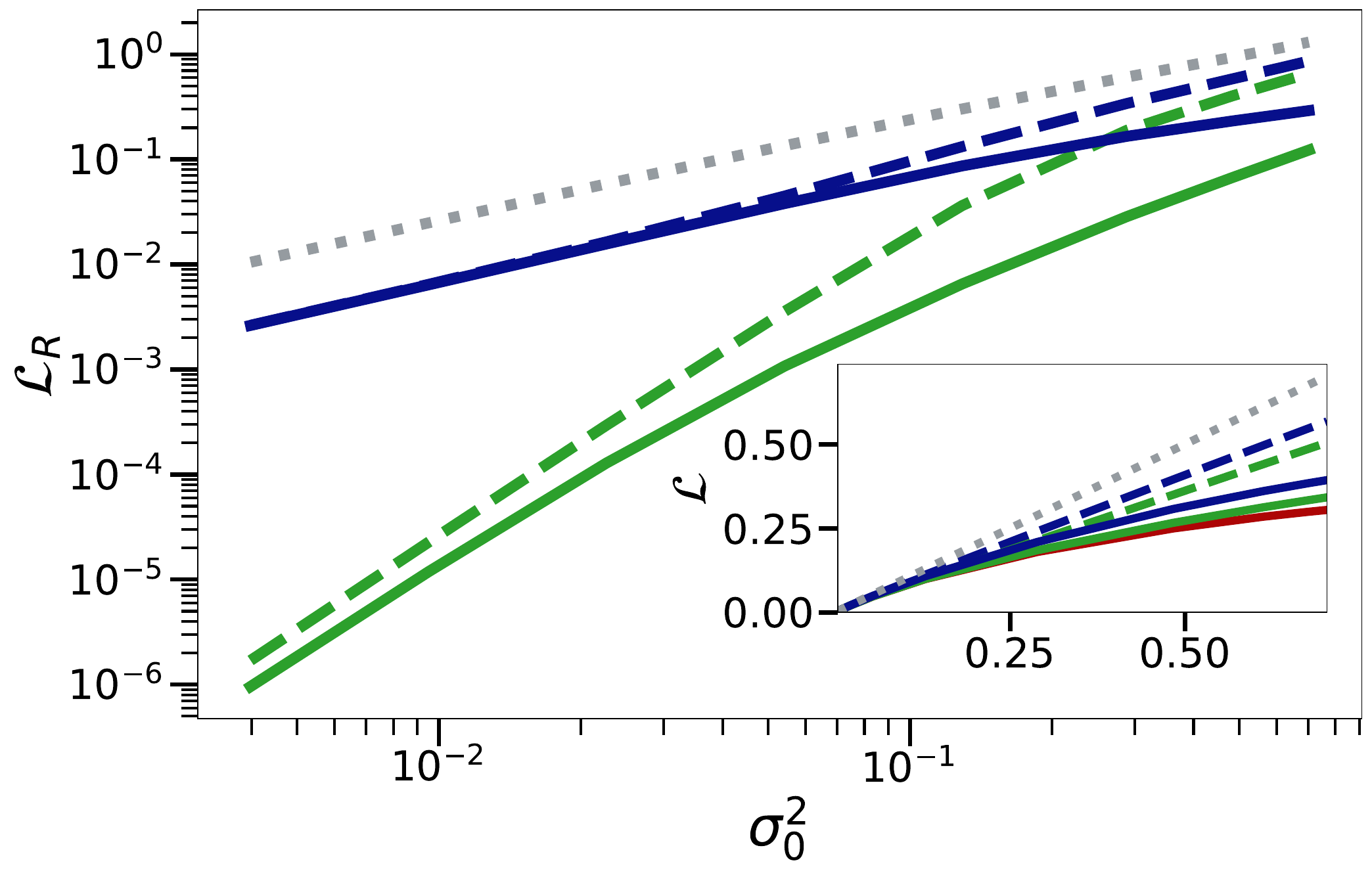}
    \caption{Relative MSL in \cref{eq:relative_MSL_figure_of_merit} as a function of the prior variance $\sigma^2_0$ for a Gaussian prior.
    Here, the task is estimating the squeezing parameter $\theta$ encoded as per \cref{eq:squeezing-transformation}, for a vacuum (top) and thermal (bottom) state $\rho_\textrm{in}$ with $\bar{n}=0.1$.
    In both plots, the dashed and solid green curves denote the relative MSL  for the $\mathcal{V}_\text{quad}$ subspace with the constrained [\cref{eq:constrained_MSL_squeezing_quadratic}] and PM estimator [\cref{eq:posterior_mean_estimator}] respectively.
    Similarly, the dashed and solid blue curves show the relative MSL of the $\mathcal{V}_\text{hom}$ subspace at homodyne angles $\phi = 0$ or $\pi/2$---chosen to minimise the MSL \cref{eq:MSL_zero_mean_quadratic_homodyne}---with the constrained [\cref{eq:MSL_zero_mean_quadratic_homodyne_minimised}] and PM estimator [\cref{eq:posterior_mean_estimator}] respectively.
    The grey dot-dashed curve indicates the \emph{a priori} MSL ratio for reference.
    Numerical values of $\mathcal{L}(\tilde{\theta}_{PM},\mathcal{P}_\mathcal{V})$ and $\mathcal{L}(\mathcal{S})$ are computed using a Fock-space truncation of 120.
    Inset: Absolute MSL $\mathcal{L}$ as a function of the \emph{a priori} width $\sigma_0^2$. 
    The red curve corresponds to the global minimum given by \cref{eq:MSL_min_global}. 
    We observe that the quadratic basis with the PM estimator (solid green) attains the lowest MSL across the full range, at the cost of a non-Gaussian PVM; the homodyne strategy with the PM estimator (solid blue) is the best directly implementable choice and overtakes the constrained estimator on the quadratic basis (dashed green) at larger $\sigma_0^2$.}
    \label{fig:squeezing_estimation_zero_mean}
\end{figure}

Despite its high rate of convergence, the full quadratic subspace $\mathcal{V}_\text{quad}$ has a practical drawback: the coefficients of $\hat{q}^2$ and $\hat{p}^2$ in $\mathcal{S}_{\mathcal{V}_\text{quad}}$ have opposite signs, which produces a spectrum that is continuous and unbounded from both above and below.
As such, the associated PVM requires non-Gaussian transformations that are experimentally challenging to implement \cite{milburn1994hyperbolic,chiribella2006optimal,oh2019optimal}, although, as noted in \cref{sec:admissible_measurements}, its statistics can be approximated by finite-dimensional Fock space truncations.

This motivates restricting to $\mathcal{V}_\text{hom} = \operatorname{span}(  \mathbbm{1}, \hat{q}_\phi^2)$, which forces both quadrature coefficients to share the same sign.
By \cref{eq:b_vanish_squeezing_linear_q}, this is the lowest-order homodyne subspace that contributes to the MSL.
Replacing $\hat{q}_\phi^2$ by the centred element $\Delta \hat{q}_\phi^2\coloneq \hat{q}_\phi^2-\langle\hat{q}_\phi^2\rangle_{\rho_0}\mathbbm{1}$ leaves the subspace $\mathcal{V}_\text{hom}$ unchanged but makes the Gram system in \cref{eq:S_constrained_opt_solutions} diagonal.
Repeating the same steps as above gives
\begin{subequations}
    \begin{align}
        \label{eq:constrained_SPM_squeezing_quadratic_homodyne}
        \mathcal{S}_{\mathcal{V}_\text{hom}}(\phi) &=\mu_0 \mathbbm{1} + \frac{b_{\Delta q_\phi^2}}{\|\Delta\hat q_\phi^2\|^2_{\rho_0}} \Delta \hat{q}_\phi^2, \\
        \label{eq:MSL_zero_mean_quadratic_homodyne}
        \mathcal{L}[\mathcal{S}_{\mathcal{V}_\text{hom}}(\phi)] &= \lambda - \mu_0^2- \frac{b_{\Delta q_\phi^2}^2}{\|\Delta\hat q_\phi^2\|^2_{\rho_0}}.
    \end{align}
\end{subequations}
This strategy corresponds to homodyne detection at angle $\phi$ with outcome $q_\phi$ post-processed by the quadratic estimator
\begin{equation}
    \label{eq:constrained_estimator_squeezing_quadratic}
    \tilde{\theta}(q_\phi) = \mu_0 + \frac{b_{\Delta q_\phi^2}}{\|\Delta\hat q_\phi^2\|^2_{\rho_0}}(q_\phi^2-\langle\hat{q}_\phi^2\rangle_{\rho_0}).
\end{equation}
Since $\Delta \hat{q}_\phi^2 = \Delta \hat{q}^2 \cos^2{\phi} + \Delta \hat{p}^2 \sin^2{\phi} +\Delta K\sin{2 \phi}$ and the $\vb{b}$ vector is linear in each basis element, $b_{\Delta q_\phi^2}$ follows from the components already obtained for $\mathcal{V}_\text{quad}$,
\footnote{
Specifically, $b_{\Delta q^2}=\tilde{V}_{qq} b_Q$, $b_{\Delta p^2}=\tilde{V}_{pp} b_P$, and $b_{\Delta K}=0$, where $b_Q=-b_P$ from Footnote~\ref{footnote:quadratic_squeezing}. For a Gaussian prior, $b_Q=-2 \sigma_0^2$.
}
giving
\begin{equation}
    b_{\Delta q_\phi^2} = -2\sigma_0^2
      (\tilde{V}_{qq}\cos^2\phi-\tilde{V}_{pp}\sin^2\phi).
\end{equation}
The norm can be calculated to be
\begin{multline}
\begin{aligned}
    \|\Delta\hat q_\phi^2\|^2_{\rho_0} =& \, \kappa_+ ( \tilde{V}_{qq}^2 \cos^4{\phi} + \tilde{V}_{pp}^2 \sin^4{\phi}) 
    \\
    &+ \frac{1}{2}(4\tilde{V}_{qp}^2 + \kappa_- \tilde{V}_{qq}\tilde{V}_{pp} )\sin^2{2\phi} 
    \\
    &+ 8 \tilde{V}_{qp}(\tilde{V}_{qq} \cos^3{\phi} \sin{\phi} + \tilde{V}_{pp} \sin^3{\phi} \cos{\phi}),
\end{aligned}
\end{multline}
with $\kappa_\pm  \coloneq 3 e^{\pm 4 \sigma_0^2} - 1$.\footnote{
Writing \[V_\phi(\theta) \coloneq \langle \hat{q}_\phi^2 \rangle_{\rho(\theta)} = V_{qq}e^{-2 \theta}\cos^2{\phi} + V_{pp}e^{2 \theta}\sin^2{\phi} + V_{qp}\sin{2\phi},\] and using $\langle \hat{q}_\phi^4 \rangle_{\rho(\theta)}=3\langle \hat{q}_\phi^2 \rangle^2_{\rho(\theta)}$ for a zero-mean Gaussian $\rho(\theta)$, one has \[\|\Delta\hat q_\phi^2\|^2_{\rho_0} = 3\int \dd{\theta} p(\theta) V_\phi(\theta)^2 - \left[\int \dd{\theta} p(\theta) V_\phi(\theta)\right]^2.\]
}
Since $\mathcal{V}_\text{hom}=\operatorname{span}(\mathbbm{1}, \hat{q}_\phi^2) \subseteq \mathcal{V}_\text{quad}$ for any $\phi$, \cref{eq:MSL_zero_mean_quadratic_homodyne} is bounded from below by that of the quadratic basis in \cref{eq:constrained_MSL_squeezing_quadratic}.

The minimum of the MSL in \cref{eq:MSL_zero_mean_quadratic_homodyne} for $V_{qp}=0$, which is satisfied by both the vacuum and thermal probes considered here, is achieved at the homodyne angles $\phi=0$ or $\phi=\pi/2$, and the MSL becomes\footnote{
Interestingly, $\mathcal{L}(\mathcal{S}_{\mathcal{V}_\text{hom}})$ is independent of the probe state at these angles.
At $\phi=0$, all terms containing $\sin{\phi}$ vanish, leaving $b_{\Delta q_0^2}=-2\sigma_0^2 \tilde{V}_{qq}$ and $\|\Delta\hat q_0^2\|^2_{\rho_0}=\kappa_+ \tilde{V}_{qq}^2$. 
The factor $\tilde{V}_{qq}^2$ then cancels, which gives $\mathcal{L}[\mathcal{S}_{\mathcal{V}_{\text{hom}}}(\phi=0)]=\sigma_0^2 - 4 \sigma_0^4/\kappa_+$.
The same rationale follows at $\phi=\pi/2$.}
\begin{equation}
    \label{eq:MSL_zero_mean_quadratic_homodyne_minimised}
    \mathcal{L}[\mathcal{S}_{\mathcal{V}_\text{hom}}(\phi=0,\pi/2)] = \sigma_0^2 - \frac{4 \sigma_0^4}{3 e^{4\sigma_0^2}-1}.
\end{equation}

The relative MSL of the quadratic and PM estimators associated with this homodyne measurement are shown in \cref{fig:squeezing_estimation_zero_mean} as dashed [\cref{eq:MSL_zero_mean_quadratic_homodyne_minimised}] and solid [\cref{eq:MSL_general_PM_estimator}] blue curves, respectively.
Both strategies again converge towards the global minimum $\mathcal{L}(\mathcal{S})$, with the PM estimator outperforming the quadratic estimator in accordance with \cref{eq:MSL_chain_of_inequalities}.
Although the full basis $\mathcal{V}_\text{quad}$ combined with the PM estimator (solid green) achieves the lowest MSL overall, it requires an experimentally challenging non-Gaussian PVM. 
By contrast, homodyne detection with the PM estimator (solid blue) remains directly implementable and, at larger prior widths, outperforms the constrained estimator on the quadratic basis (dashed green) despite spanning a smaller subspace.
This does not, however, violate \cref{eq:MSL_chain_of_inequalities}, since these correspond to different subspaces.

\section{Discussion}
\label{sec:discussion}

We have developed a general framework in which the optimisation of the MSL over measurements and estimators for a CV system can be reduced to a finite linear system of equations in \cref{eq:S_constrained_opt_solutions} once the Hermitian operator encoding the measurement-estimator pair is constrained to a chosen operator subspace $\mathcal{V}=\operatorname{span}(B_1,\ldots,B_d)$.
For Gaussian states, a natural choice of basis $B_i$ consists of polynomials in the quadratures, whose moments are computable directly from the first and second moments of the state $\rho(\theta)$.
The solution to this linear system defines a projected SPM operator $\mathcal{S}_\mathcal{V}$, which is the orthogonal projection of the global optimum $\mathcal{S}$ onto the subspace $\mathcal{V}$ with respect to the $\rho_0$-weighted inner product in \cref{eq:inner_product_rho0}, and, like $\mathcal{S}$, defines a constrained strategy via its spectral decomposition.
The excess MSL $\mathcal{L}(\mathcal{S}_\mathcal{V})-\mathcal{L}(\mathcal{S})$ incurred by the constrained strategy is given by the squared distance $\norm{\mathcal{S}_\mathcal{V} - \mathcal{S}}^2_{\rho_0}$ in \cref{eq:absolute-MSL-difference}, and, for faithful $\rho_0$, \cref{theorem:1} identifies the polynomial subspaces $\mathcal{V}$ for which it vanishes.
Retaining the spectral measurement defined by $\mathcal{S}_\mathcal{V}$ but replacing the corresponding estimator with the PM can only reduce, or leave unchanged, the MSL at no additional experimental cost, with the same squared distance then serving as an upper bound on the excess.

Since closed-form expressions for the MSL are often difficult to obtain, this framework provides one of the few analytic toolboxes currently available for CV Bayesian estimation.
We expect it to be useful not only as a standalone method, but also as a benchmark or subroutine within numerical approaches such as those of Refs.~\cite{bavaresco2024designing,andre2026strategy}.

Our worked examples show that, depending on the choice of subspace $\mathcal{V}$, the resulting schemes are either exactly optimal or near-optimal.
For displacement estimation with a Gaussian prior, $\rho_0$ is Gaussian and \cref{theorem:1} implies that homodyne detection is exactly optimal for any prior width, recovering the result of Ref.~\cite{morelli2021bayesian} by a different route.
In contrast, for a uniform prior the state $\rho_0$ is no longer Gaussian, so \cref{theorem:1} no longer applies and the linear shrinkage estimator ceases to be optimal.
Nevertheless, it remains close to optimal, while including higher polynomial basis elements provides a systematic and tractable approximation to the Bayesian optimum.

For squeezing estimation of an undisplaced probe, the parameter is encoded solely in the covariance matrix.
Consequently, all moments $\mathrm{Tr}(\bar{\rho}\,\hat{q}_\phi^{2n+1})$ vanish for $n \geq 0$, so any $f$-transformed estimator odd in the homodyne outcome $q_\phi$ does not contribute to the MSL in \cref{eq:MSL_constrained_diagonal}.
For the full quadratic subspace $\mathcal{V}_\text{quad}$, we obtained closed-form expressions for the projected SPM operator $\mathcal{S}_\mathcal{V}$ and the corresponding MSL $\mathcal{L}(\mathcal{S}_\mathcal{V})$ [\cref{eq:constrained_SPM_squeezing_quadratic,eq:constrained_MSL_squeezing_quadratic}], although the associated PVM involves non-normalisable scattering states and is experimentally challenging to implement.
Restricting instead to the subspace $\mathcal{V}_\textrm{hom}$, corresponding to homodyne detection at angle $\phi$ with an estimator quadratic in the outcome, yields closed-form expressions for both the estimator \cref{eq:constrained_estimator_squeezing_quadratic} and the MSL \cref{eq:MSL_zero_mean_quadratic_homodyne}.
While the quadratic basis achieves a lower MSL throughout, the homodyne strategy remains directly implementable and provides a practically accessible near-optimal scheme.
In all cases, the strategies converge rapidly to the global optimum as the prior narrows, with the PM estimator providing a larger improvement at broader prior widths.

Several limitations are worth noting.
First, the linear system $G \vb*{\alpha}=\vb{b}$ requires the traces $\mathrm{Tr}(\rho_0 B_i B_j)$ and $\mathrm{Tr}(\bar{\rho} B_i)$ to be finite.
This is satisfied for polynomial bases and Gaussian states with standard priors, but must be checked independently in more general settings.
Second, the spectral decomposition of $\mathcal{S}_\mathcal{V}$ defines a physical projective measurement only when $\mathcal{S}_\mathcal{V}$ is self-adjoint.
While this is true for all examples considered here---including real polynomial bases in a single quadrature and real quadratic forms in the canonical quadratures---it is not guaranteed for higher-order bases that mix the quadratures, and must therefore be verified case by case.
Nevertheless, even when a spectral decomposition is unavailable, $\mathcal{L}(\mathcal{S}_\mathcal{V})$ in \cref{eq:MSL_constrained} remains a valid lower bound on the MSL achievable by any Hermitian operator within $\mathcal{V}$, although the bound need not be attainable.

Looking ahead, several extensions widen the scope of the framework.
Within the Gaussian setting, it accommodates any estimation scenario generated by Gaussian unitaries acting on Gaussian states, since these act as symplectic transformations on the first and second moments of the state as $\overline{\vb{r}} \mapsto \Sigma \overline{\vb{r}}$ and $V \mapsto \Sigma V \Sigma^\top$, with $\Sigma$ symplectic \cite{adesso2014continuous,serafini2017quantum,weedbrook2012gaussian}.
This covers, for instance, phase and squeezing estimation with displaced probes, as well as loss estimation and the characterisation of Gaussian noise and dissipation \cite{monras2007optimal,zhou2023bayesian,serafini2005quantifying,adesso2009optimal}.

More importantly, the framework is \emph{not} restricted to Gaussian states.
The only requirements are that $\rho_0$ and $\bar{\rho}$ are well defined and that the relevant operator moments exist, so that the Gram matrix $G$ and vector $\vb{b}$ in \cref{eq:G_and_b_def} are finite.
Since these conditions make no reference to the dimension of the Hilbert space, the optimisation applies equally to non-Gaussian CV probes and, with a suitable operator basis, to finite-dimensional systems.

Correspondingly, the choice of basis is open, and while polynomials in the quadratures are natural for Gaussian states, other Hermitian bases may be better suited to the problem at hand.
For a qubit, the identity together with the Pauli operators span the full space of Hermitian operators, so that $\mathcal{S}_\mathcal{V}=\mathcal{S}$ exactly.
In CV systems, bases built from Fock states offer an alternative to the quadratures; e.g., the diagonal subspace $\operatorname{span}\{\ketbra{j}\}_{j \leq n}$ describes photon counting up to $n$ quanta, equipped with its optimal Bayes estimator \cref{eq:posterior_mean_estimator}.

Finally, for Gaussian probes and polynomial bases, $\mathcal L(\mathcal S_\mathcal V)$ is an explicit function of the (prior-averaged) first and second moments of the state $\rho(\theta)$.
Probe optimisation can therefore be carried out directly within the same framework, providing a fully Bayesian analogue of state optimisation in local QET.
Another natural direction is multiparameter estimation, where the relationship between subspace choice and measurement incompatibility \cite{albarelli2025measurement} merits investigation.

\section*{Acknowledgements}

We gratefully thank L. Correa, D. Branford, F. Albarelli, R. Rodr\'{i}guez, and S. Nimmrichter  for insightful discussions. 
We acknowledge funding by the Ministerio de Ciencia e Innovaci\'{o}n and European Union (FEDER) (PID2022-138269NB-I00). 
E.G. is funded by the Engineering and Physical Sciences Research Council (Grant No. EP/W524451/1).
E. G. also acknowledges the University of La Laguna (ULL), and the Spanish Ministry of Universities for supporting the `QSB Visiting Program', during which this project started.
J. R. acknowledges financial support from the Surrey Future Fellowship (University of Surrey), during which this project was developed, and from the UKRI-JST project ``Quantum Control \& Sensing: Enhancing high-precision quantum sensing in noisy environments via optimised control''.

\section*{Data and code availability}
The code used to generate the numerical results and figures in this work, together with the plotted data, is openly available at \cite{gandar2026gaussian} under the MIT licence.
The repository also includes the symbolic computations verifying the analytical expressions reported here.
No experimental data were generated.

% References
\bibliographystyle{apsk}
\bibliography{refs}

\appendix

\section{Variational derivation of the constrained optimum}
\label[appendix]{appendix:alternate_proof_of_constrained_optimum}

\cref{eq:S_constrained_opt_solutions} may alternatively be derived via a variational argument.
Recall that minimising the MSL in \cref{eq:average_info_error_general_POVM} is equivalent to minimising the functional [\cref{eq:J-functional}]
\begin{equation}
    \label{eq:}
    J(\mathcal{M}_{1}) = \mathrm{Tr}(\rho_0 \mathcal{M}^2_{1}) - 2\mathrm{Tr}(\bar{\rho} \mathcal{M}_{1}),
\end{equation}
where \(\mathcal{M}_{1} \in \mathcal{V}\). 
For any $X \in \mathcal{V}$, the variation $\mathcal{M}_{1} + \varepsilon X$ gives
\begin{equation}
    \label{eq:appendix_MSL_variation}
    \begin{aligned}
        J(\mathcal{M}_{1} + \varepsilon X)  
        =\,& J(\mathcal{M}_{1}) 
        + \varepsilon [\mathrm{Tr}(\rho_0 \{\mathcal{M}_{1},X\})
        - 2 \mathrm{Tr}(\bar{\rho} X)]
        \\
        & 
        + \varepsilon^2 \mathrm{Tr}(\rho_0 X^2).
    \end{aligned}
\end{equation}
Now, a necessary condition for \(\mathcal{M}_{1}=\mathcal{S}_{\mathcal{V}}\) to be a stationary point of \(J\) within \(\mathcal{V}\) is
\begin{equation}
    \dv{}{\varepsilon} \eval{J(\mathcal{S}_\mathcal{V} + \varepsilon X) }_{\varepsilon = 0} =0,
    \quad 
    \forall X \in \mathcal{V}.
\end{equation}
This condition takes the form
\begin{equation}
    \label{eq:appendix_stationarity_condition}
     \mathrm{Tr}(\rho_0 \{\mathcal{S}_\mathcal{V},X\}) - 2 \mathrm{Tr}(\bar{\rho} X) = 0, \quad \forall X \in \mathcal{V}.
\end{equation}
Then, choosing the variations as the basis elements $X = B_i$ and expanding $\mathcal{S}_\mathcal{V}$ in this basis as
\begin{equation}
    \mathcal{S}_\mathcal{V} = \sum_j \alpha_j B_j,
\end{equation}
one finally obtains
\begin{equation}
    \sum_j \mathrm{Tr}(\rho_0 \{B_i, B_j\}) \alpha_j = 2 \mathrm{Tr}(\bar{\rho} B_i),
\end{equation}
which is the same set of linear equations in \cref{eq:S_constrained_opt_solutions}.

It remains to show that $\mathcal{S}_\mathcal{V}$ is a minimum rather than a stationary point.
The stationarity condition \cref{eq:appendix_stationarity_condition} forces the linear term in $J(\mathcal{S}_\mathcal{V} + \varepsilon X)$ to vanish, leaving
\begin{equation}
    J(\mathcal{S}_\mathcal{V} + \varepsilon X) = J(\mathcal{S}_\mathcal{V}) + \varepsilon^2 \mathrm{Tr}(\rho_0 X^2) \geq J(\mathcal{S}_\mathcal{V}),
\end{equation}
where the inequality holds for all $\varepsilon \in \mathbbm{R}$, since $\rho_0 \succeq 0$ implies $\mathrm{Tr}(\rho_0 X^2) \geq 0$.
As every element of $\mathcal{V}$ can be written as $\mathcal{S}_\mathcal{V} + \varepsilon X$ for some $X \in \mathcal{V}$, this confirms $\mathcal{S}_\mathcal{V}$ is a global minimum of $J$, and thus the full MSL, within the subspace $\mathcal{V}$.

This minimum is unique iff the Gram matrix $G$, defined in Eq.~\eqref{eq:gram-norm-version}, is positive definite.
To see this, write $X = \sum_i \alpha_i B_i$ for general coefficients $\vb*{\alpha} = (\alpha_1, \alpha_2, \dots)$ and note that
\begin{align}
    \mathrm{Tr}(\rho_0 X^2) 
    &= \sum_{ij} \alpha_i \alpha_j \mathrm{Tr}(\rho_0 B_i B_j)
    \\ 
    &= \frac{1}{2} \sum_{ij} \alpha_i \alpha_j \mathrm{Tr}(\rho_0 \{B_i, B_j\})
    \\
    &= \vb*{\alpha}^\top G \vb*{\alpha},
\end{align}
where the second equality uses symmetry of $\alpha_i \alpha_j$ under $i \leftrightarrow j$.
Hence, $\mathrm{Tr}(\rho_0 X^2) >0$ for all non-zero $X \in \mathcal{V}$ iff $G$ is positive definite.
If $G$ is singular, then there exists a non-zero $X$ with $\mathrm{Tr}(\rho_0 X^2)=0$, so $\mathcal{S}_\mathcal{V} + \varepsilon X$ is also a minimiser for any $\varepsilon \in \mathbbm{R}$.
In this case, the minimiser is not unique but one may select the minimum-norm solution via the Moore-Penrose pseudoinverse.

\section{Proof of Theorem 1}
\label[appendix]{appendix:polynomial_S_proof}
We restate \cref{theorem:1} here:
\textit{Let $\rho_0$ be a faithful Gaussian state.
Then $\mathcal{S}$ is a polynomial of degree $R$ in the quadrature operators if and only if $W_{\bar{\rho}}/W_{\rho_0}$ is a polynomial of degree $R$ in the phase-space coordinates $\vb{r}$.}

One way to prove this theorem is to work in phase space.
To do so, 
we consider the Wigner-Weyl transform \cite{curtright2014concise,folland2016harmonic,hall2013quantum,de2011symplectic}, which is an invertible mapping between operators $X$ acting on a Hilbert space $\mathcal{H}$ and functions $W_X(\vb{r})$, known as Weyl symbols, where $\vb{r} \in \mathbbm{R}^{2N}$ are phase-space coordinates.
For $X=\rho$, $W_\rho(\vb{r})$ is the Wigner function.
Furthermore, a polynomial operator of degree $R$ in the quadrature operators corresponds to a polynomial Weyl symbol of degree $R$ in phase-space coordinates. Conversely, every polynomial Weyl symbol defines a unique polynomial operator through Weyl quantisation (see \cite[Ch.~13]{hall2013quantum} and \cite[Prop.~2.11]{folland2016harmonic}).

The Weyl symbol of a product $A B$ is given by the Moyal (star) product of the symbols.
For our purposes, the most useful way of expressing this is the formal series \cite{curtright2014concise,moyal1949quantum,hillery1984distribution,groenewold1946principles,szabo2003quantum}
\begin{equation}
    \label{eq:moyal}
    \begin{aligned}
    W_{AB}(\vb{r}) =&\, (W_A \star W_B)(\vb{r}) \\
    =&\,  \sum_{m=0}^{\infty} \frac{1}{m!}\left(\frac{i\hbar}{2}\right)^m \Omega^{i_1 j_1}\cdots\Omega^{i_m j_m}\\ 
    &\times (\partial_{i_1}\cdots\partial_{i_m} W_A)(\partial_{j_1}\cdots\partial_{j_m} W_B),
    \end{aligned}
\end{equation}
where $\Omega$ is the symplectic form and $\partial_{i}:=\partial/\partial r_i$.
Here and in the following, we use Einstein summation notation throughout, and we set $\hbar=1$.

Applying the Weyl transform to the Lyapunov equation $\mathcal{S} \rho_0 + \rho_0 \mathcal{S} = 2 \bar{\rho}$ gives the equivalent phase-space equation
\begin{equation}
    \label{eq:lyapunov_moyal}
    W_{\mathcal{S}} \star W_{\rho_0} + W_{\rho_0} \star W_{\mathcal{S}} = 2 W_{\bar{\rho}}.
\end{equation}
Defining the map
\begin{equation}
    \label{eq:L_operator_def}
    \Lambda\left[g\right] := \frac{g \star W_{\rho_0} + W_{\rho_0} \star g}{2},
\end{equation}
acting on phase-space functions $g(\vb{r})$, \cref{eq:lyapunov_moyal} can be written as
\begin{equation}
    \Lambda\left[W_{\mathcal{S}}\right] = W_{\bar{\rho}} = \Phi W_{\rho_0},
\end{equation}
where $\Phi := W_{\bar{\rho}}/W_{\rho_0}$.

Only even-order terms in the Moyal series \cref{eq:moyal} survive in $\Lambda$ because the antisymmetry of $\Omega$ kills the odd-order contributions under the exchange of arguments. 
With $\gamma_m := (-1)^m/[(2m)!\,4^m]$, the map \cref{eq:L_operator_def} becomes
\begin{equation}
    \begin{aligned}
        \label{eq:appendix_weyl_lyapunov_full}
        \Lambda\left[g\right] = W_{\rho_0} \sum_{m= 0}^{\infty} \gamma_m  \Omega^{i_1 j_1} \cdots \Omega^{i_{2m} j_{2m}} \\ \times \left(\partial_{i_1} \cdots \partial_{i_{2m}} g\right)\left(\frac{\partial_{j_1} \cdots \partial_{j_{2m}} W_{\rho_0}}{W_{\rho_0}}\right).
    \end{aligned}
\end{equation}

The key technical result is that $\Lambda$ preserves polynomial degree against a Gaussian $W_{\rho_0}(\vb{r})$.

Let $P_R$ denote the space of polynomials of degree \emph{at most} $R$ in the phase-space coordinates $\vb{r} \in \mathbbm{R}^{2N}$, and define $F_R =  \left\{ P W_{\rho_0} | P \in P_R \right\}$ as the space of polynomials of degree \emph{at most} $R$ multiplied by $W_{\rho_0}$.

\begin{flushleft}
    \textbf{Lemma:} If $W_{\rho_0}$ is Gaussian, then $\Lambda$ maps $P_R$ into $F_R$. 
\end{flushleft}
\textbf{Proof:} Writing the Gaussian Weyl symbol of $\rho_0$ as
\begin{equation}
    W_{\rho_0}(\vb{r})=\frac{1}{(2\pi)^N \sqrt{\det V_0}} \exp \left(-\frac{1}{2}(\vb{r}-\overline{\vb{r}}_0)^{T} V_0^{-1}(\vb{r}-\overline{\vb{r}}_0)\right),
\end{equation}
the first two  derivatives
\begin{equation}
    \begin{aligned}
        \frac{\partial_i W_{\rho_0}}{W_{\rho_0}} &= -(V_0^{-1}(\vb{r}-\overline{\vb{r}}_0))_i, \\
        \frac{\partial_i \partial_j W_{\rho_0}}{W_{\rho_0}} &=(V_0^{-1}(\vb{r}-\overline{\vb{r}}_0))_i (V_0^{-1}(\vb{r}-\overline{\vb{r}}_0))_j - (V_0^{-1})_{ij},
    \end{aligned}
\end{equation}
are real polynomials in $\vb{r}$ of degree 1 and 2, respectively.

Applying the product rule inductively, we observe that $(\partial_{j_1} \cdots \partial_{j_{2m}} W_{\rho_0})/W_{\rho_0} \in P_{2m}$ for all $m \geq 0$.
For $g \in P_R$, the derivative $\partial_{i_1} \cdots \partial_{i_{2m}} g \in P_{R-2m}$ when $2m \leq R$ and vanishes otherwise, so that the series terminates.
As such, each term in \cref{eq:appendix_weyl_lyapunov_full} contributes a polynomial of total degree at most $(R-2m)+2m=R$, so $\Lambda\left[g\right] \in F_R$. \qed \\

We are now in a position to prove the two directions of \cref{theorem:1}.

\subsection{Forward direction}
Suppose $\mathcal{S}$ is a degree-$R$ polynomial in the quadratures, so $W_{\mathcal{S}} \in P_R$.
The Lemma then gives $\Lambda\left[W_{\mathcal{S}}\right] \in F_R$, and since $\Lambda\left[W_{\mathcal{S}}\right] = W_{\bar{\rho}}$, it follows that $\Phi := W_{\bar{\rho}}/W_{\rho_0} \in P_R$. \qed

\subsection{Converse direction}
Suppose that $\Phi = W_{\bar{\rho}}/W_{\rho_0} \in P_R$, or equivalently $W_{\bar{\rho}} \in F_R$.
We wish to show that $W_{\mathcal{S}} \in P_R$, which implies that $\mathcal{S}$ is a degree-$R$ polynomial.

We claim that the restriction $\Lambda|_{P_R} : P_R \to F_R$, obtained by restricting the domain of $\Lambda$ to $P_R$, is a bijective map.
This would imply, since $W_{\bar{\rho}} \in F_R$ by hypothesis, that there exists a unique polynomial Weyl symbol $W_\text{poly} \in P_R$ such that $\Lambda[W_\text{poly}]=W_{\bar{\rho}}$.
We will then verify that the corresponding operator $S_\text{poly}$ coincides with the SPM operator $\mathcal{S}$.

To start, note that, for a Gaussian $\rho_0$, $W_{\rho_0}(\vb{r}) \neq 0$ everywhere, so multiplication by $W_{\rho_0}$ is an isomorphism $P_R \to F_R$.
Furthermore, the dimensions of these spaces are equal and finite with $\operatorname{dim}(P_R)=\operatorname{dim}(F_R) = \binom{2N+R}{R} < \infty$.

In order to show that $\Lambda|_{P_R}$ is injective,
suppose $W_\text{poly} \in P_R$ satisfies $\Lambda[W_\text{poly}]=0$.
This means that the operator $S_\text{poly}$ with Weyl symbol $W_\text{poly}$ satisfies
\begin{equation}
    S_\text{poly} \rho_0 + \rho_0 S_\text{poly}=0.
\end{equation}
Sandwiching between eigenvectors $\{\ket{\psi_i}\}$ of $\rho_0$ gives
\begin{equation}
    (\lambda_i + \lambda_j)\mel{\psi_i}{S_\text{poly}}{\psi_j} = 0 \quad \forall i,j,
\end{equation}
where $\lambda_i$ denotes the corresponding eigenvalues of $\rho_0$.
The matrix elements above are well defined because the eigenstates $\ket{\psi_j}$ of a Gaussian $\rho_0$ are (possibly displaced and squeezed) Fock states, whose wavefunctions are Schwartz-class \cite[Thm.~11.4]{hall2013quantum} and lie in the domain of any polynomial operator in $\hat q,\hat p$.\footnote{The Schwartz class, which is the space of smooth, rapidly decaying functions, is preserved by polynomial operators in $\hat{q},\hat{p}$ and by displacements \cite[Sec.~1.3]{folland2016harmonic} and squeezing (a symplectic transformation) \cite[Prop.~4.27]{folland2016harmonic}.
Therefore $\ket{\psi_i} \in \operatorname{Dom}(S_\text{poly})$ for any polynomial operator, and $\mel{\psi_i}{S_\text{poly}}{\psi_j}$ is well-defined.}

Since $\rho_0$ is faithful, meaning $\operatorname{ker}(\rho_0)=\{0\}$, and because $\rho_0$ is positive semi-definite, then all $\lambda_i>0$.
Therefore $\lambda_i + \lambda_j >0$ forces $\mel{\psi_i}{S_\text{poly}}{\psi_j}=0$ for all $i,j$.
Since the eigenvectors $\{\ket{\psi_i}\}$ form a complete orthonormal basis for all of $\mathcal{H}$, this implies $S_\text{poly}=0$ and hence $W_\text{poly}=0$, so $\Lambda|_{P_R}$ is injective.

We now use this to show existence of a polynomial solution.
Since $\Lambda|_{P_R}$ is an injective linear map between finite-dimensional spaces of equal dimension, it is bijective.
Because $W_{\bar{\rho}} \in F_R$  by hypothesis, there exists a unique $W_\text{poly} \in P_R$ satisfying
\begin{equation}
    \Lambda[W_\text{poly}] = W_{\bar{\rho}}.
\end{equation}
Equivalently, the operator $S_\text{poly}$ is a polynomial of degree at most $R$ in the quadratures that solves the Lyapunov equation in \cref{eq:lyapunov}.

It remains to show that the SPM operator $\mathcal{S}$ coincides with $S_\text{poly}$.
Both satisfy the same Lyapunov equation $X \rho_0 + \rho_0 X = 2 \bar{\rho}$, so taking matrix elements in the $\{\ket{\psi_i}\}$ basis of $\rho_0$ gives
\begin{equation}
    (\lambda_i + \lambda_j)\mel{\psi_i}{X}{\psi_j} = 2\mel{\psi_i}{\bar{\rho}}{\psi_j}.
\end{equation}
Since $\rho_0$ is faithful, all $\lambda_i$ > 0 and therefore
\begin{equation}
    \label{eq:appendix:SPM_matrix_elements}
    \mel{\psi_i}{X}{\psi_j} = \frac{2\mel{\psi_i}{\bar{\rho}}{\psi_j}}{\lambda_i + \lambda_j}
\end{equation}
is uniquely determined for all $i,j$.
Applying this both to $X=\mathcal{S}$ and $X=S_\text{poly}$ yields
\begin{equation}
    \mel{\psi_i}{\mathcal{S}}{\psi_j} = \mel{\psi_i}{S_\text{poly}}{\psi_j}.
\end{equation}
Since $\{\ket{\psi_i}\}$ is a complete orthonormal basis for $\mathcal{H}$, this implies $\mathcal{S}=S_\text{poly}$ as operators on $\mathcal{H}$. \qed \\

\subsection{Exact-degree statement}
The forward and converse directions above establish that, for faithful Gaussian $\rho_0$, $W_{\mathcal{S}} \in P_R$ if and only if $\Phi \in P_R$, i.e. $\mathcal{S}$ is a polynomial of degree at most $R$ in the quadratures if and only if $\Phi$ is a polynomial of degree at most $R$ in the phase-space coordinates $\vb{r}$.
We now sharpen this to the exact-degree statement of \cref{theorem:1} using proof by contradiction.

Suppose $\mathcal{S}$ has degree exactly $R$, i.e., $W_{\mathcal{S}} \in P_R$ but $W_{\mathcal{S}} \notin P_{R-1}$. 
The forward direction gives $\Phi \in P_R$.
If $\Phi$ had a lower degree, i.e., $\Phi \in P_{R-1}$, the converse direction applied at degree $R-1$ would yield $W_{\mathcal{S}} \in P_{R-1}$, contradicting $W_{\mathcal{S}} \notin P_{R-1}$.
Hence $\Phi \notin P_{R-1}$, so $\Phi$ has degree exactly $R$.

On the other hand, suppose that $\Phi$ has degree exactly $R$.
The converse direction gives $W_{\mathcal{S}} \in P_R$.
If $W_{\mathcal{S}} \in P_{R-1}$, the forward direction applied at degree $R-1$ would yield $\Phi \in P_{R-1}$, contradicting $\Phi \notin P_{R-1}$.
Hence $W_{\mathcal{S}} \notin P_{R-1}$, so $\mathcal{S}$ has degree exactly $R$. \qed \\

Lastly, we note that $\rho_0$ is faithful in two practically relevant cases: (i) if each $\rho(\theta)$ is itself faithful (e.g., thermal probe states), then $\rho_0$ is faithful regardless of the prior; (ii) if $\rho(\theta)=\ketbra{\psi_\theta}$ are pure probe states, then $\rho_0$ is faithful provided that the vectors $\{\ket{\psi_\theta}\}$ span the full Hilbert space over the support of $p(\theta)$, meaning that there is no non-zero $\ket{\phi}$ such that $\braket{\phi}{\psi_\theta}=0 \; \forall \theta$ in the support of the prior.
This occurs, for example, in displacement estimation with coherent or vacuum probes and a Gaussian prior, since coherent states form an over-complete basis.

\section{Excess MSL and geometric interpretation of projected SPM operators}
\label[appendix]{appendix:proof_average_info_error_quadratic_inequality}

This appendix proves \cref{eq:average_info_error_constrained_difference_equality} and subsequently shows that $\mathcal{S}_\mathcal{V}$ is the orthogonal projection of the global SPM operator $\mathcal{S}$ onto the subspace $\mathcal{V}$.

For any Hermitian $X$, as per \cref{eq:lyapunov}, one may write
\begin{equation}
    \label{eq:lyapunov_rho0_inner_product}
    \expval{\mathcal{S},X}_{\rho_0} = \mathrm{Tr}(\bar{\rho} X ),
\end{equation}
where $\mathcal{S}$ denotes the SPM operator and the $\rho_0$-weighted inner product is defined in \cref{eq:inner_product_rho0}.
Accordingly, the functional \(J(X)\) defined in \cref{eq:J-functional} can be written as
\begin{equation}
    \label{eq:J-functional-app}
    J(X) = \expval{X,X}_{\rho_0} - 2\expval{\mathcal{S},X}_{\rho_0},
\end{equation}
and its evaluation at $\mathcal{S}$ as
\begin{equation}
    J(\mathcal{S}) = -\expval{\mathcal{S},\mathcal{S}}_{\rho_0}.
\end{equation}
Further observe that, since the inner product is linear in both arguments,
\begin{equation}
    \label{eq:aux-identity-app}
    \expval{X-\mathcal{S},X-\mathcal{S}}_{\rho_0} = \expval{X,X}_{\rho_0} - 2\expval{\mathcal{S},X}_{\rho_0} + \expval{\mathcal{S},\mathcal{S}}_{\rho_0}.
\end{equation}
Combining Eqs.~(\ref{eq:J-functional-app}--\ref{eq:aux-identity-app}) leads to
\begin{equation}
    J(X) = \expval{X-\mathcal{S},X-\mathcal{S}}_{\rho_0} + J(\mathcal{S}).
\end{equation}
Given that the MSL for the Hermitian \(X\) may be written as \(\mathcal{L}(X) = \lambda + J(X)\), the excess MSL is calculated as
\begin{align}
    \mathcal{L}(X) - \mathcal{L}(\mathcal{S}) &= \expval{X-\mathcal{S},X-\mathcal{S}}_{\rho_0}
    \\
    &= \norm{X-\mathcal{S}}^2_{\rho_0} \geq 0,
\end{align}
with equality iff $X-\mathcal{S}$ vanishes on the support of $\rho_0$.
This proves \cref{eq:average_info_error_constrained_difference_equality} in the main text.

We next justify why 
\begin{equation}
    \mathcal{S}_\mathcal{V} = \underset{X \in \mathcal{V}}{\text{argmin}} \: \norm{X-\mathcal{S}}^2_{\rho_0}
\end{equation}
corresponds to a projection of the SPM operator $\mathcal{S}$ onto the subspace $\mathcal{V}$.
First, rewrite the stationary condition in \cref{eq:appendix_stationarity_condition} as $\expval{\mathcal{S}_\mathcal{V},X}_{\rho_0}=\mathrm{Tr}(\bar{\rho} X)$.
Subtracting \cref{eq:lyapunov_rho0_inner_product} gives
\begin{equation}
    \label{eq:appendix_stationarity_condition2}
    \expval{\mathcal{S}-\mathcal{S}_\mathcal{V},X}_{\rho_0} = 0 \quad \forall X \in \mathcal{V},
\end{equation}
i.e., $\mathcal{S}-\mathcal{S}_\mathcal{V}$ is orthogonal to all of $\mathcal{V}$.
This is the characteristic property of an orthogonal projection \cite{kreyszig1991introductory}.

To see uniqueness, suppose both $\mathcal{S}_\mathcal{V}$ and $\mathcal{S}_\mathcal{V}'$ are elements of $\mathcal{V}$ with $\mathcal{S}-\mathcal{S}_\mathcal{V}$ and $\mathcal{S}-\mathcal{S}_\mathcal{V}'$ both orthogonal to all of $\mathcal{V}$.
Then their difference $\mathcal{S}_\mathcal{V}-\mathcal{S}_\mathcal{V}' \in \mathcal{V}$ satisfies
\begin{equation}
    \expval{\mathcal{S}_\mathcal{V}-\mathcal{S}_\mathcal{V}',X}_{\rho_0}=\expval{(\mathcal{S}-\mathcal{S}_\mathcal{V}')-(\mathcal{S}-\mathcal{S}_\mathcal{V}),X}_{\rho_0} =0,
\end{equation}
for all $X \in \mathcal{V}$.
So $\mathcal{S}_\mathcal{V}-\mathcal{S}_\mathcal{V}'$ is orthogonal to all of $\mathcal{V}$.
In particular, it is orthogonal to itself, and $||\mathcal{S}_\mathcal{V}-\mathcal{S}_\mathcal{V}'||^2_{\rho_0}=0$ implies that $\mathcal{S}_\mathcal{V}=\mathcal{S}_\mathcal{V}'$ on the support of $\rho_0$.

Since $\mathcal{S}_\mathcal{V} \in \mathcal{V}$ is the unique element satisfying the condition in \cref{eq:appendix_stationarity_condition2}, we may write $\mathcal{S}_\mathcal{V}=\Pi_\mathcal{V}^{(\rho_0)}\mathcal{S}$ with projector $\Pi_\mathcal{V}^{(\rho_0)}$.
This gives the decomposition 
\begin{equation}
    \mathcal{S}=\Pi_\mathcal{V}^{(\rho_0)}\mathcal{S}+\big[\mathbbm{1}-\Pi_\mathcal{V}^{(\rho_0)}\big]\mathcal{S},    
\end{equation}
where the two parts are orthogonal w.r.t the $\rho_0$-weighted inner product 
\begin{equation}
    \big\langle\Pi_\mathcal{V}^{(\rho_0)}\mathcal{S},\big[\mathbbm{1}-\Pi_\mathcal{V}^{(\rho_0)}\big]\mathcal{S}\big\rangle_{\rho_0}=\expval{\mathcal{S}_\mathcal{V},\mathcal{S}-\mathcal{S}_\mathcal{V}}_{\rho_0}=0,
\end{equation}
taking $X=\mathcal{S}_\mathcal{V} \in \mathcal{V}$ in \cref{eq:appendix_stationarity_condition2}.

\end{document}